\documentclass[letterpaper,autoref]{lipics-v2021}

\nolinenumbers

\hideLIPIcs

\usepackage{amsmath,amssymb,amsfonts,latexsym}
\usepackage{color,graphicx}
\usepackage{url} 
\usepackage{fancybox,framed}
\usepackage{algorithm}
\usepackage[noend]{algpseudocode}
\usepackage{mathtools,thmtools}
\usepackage{enumerate,xspace}
\usepackage{siunitx} 
\usepackage[commandnameprefix=ifneeded]{changes}

\newcommand{\myomit}[1]{}

\newcommand{\eps}{\varepsilon}
\renewcommand{\epsilon}{\eps}
\newcommand{\etal}{\emph{et al.}\xspace}

\theoremstyle{plain}

\newenvironment{myquote}%
  {\list{}{\leftmargin=4mm\rightmargin=4mm}\item[]}%
  {\endlist}

\newcommand{\A}{\ensuremath{\mathcal{A}}}
\newcommand{\B}{\ensuremath{\mathcal{B}}}
\newcommand{\C}{\ensuremath{\mathcal{C}}}
\newcommand{\D}{\ensuremath{\mathcal{D}}}

\newcommand{\G}{\ensuremath{\mathcal{G}}}
\newcommand{\ig}{\G^{\times}}               

\newcommand{\cP}{\ensuremath{\mathcal{P}}}

\newcommand{\R}{\ensuremath{\mathcal{R}}}
\newcommand{\cS}{\ensuremath{\mathcal{S}}}

\newcommand{\REAL}{\ensuremath{\mathbb{R}}}
\newcommand{\Reals}{\REAL}

\newcommand{\Ex}[1]{\ensuremath{\mathbf{E}[#1]}}
\renewcommand{\Pr}[1]{\ensuremath{\mathbf{Pr}[#1]}}
\renewcommand{\leq}{\leqslant}
\renewcommand{\geq}{\geqslant}
\newcommand{\bd}{\partial}
\newcommand{\intr}{\mathit{int}}
\DeclareMathOperator{\polylog}{polylog}

\DeclarePairedDelimiter\floor{\lfloor}{\rfloor}

\DeclareMathOperator{\diam}{diam}

\newcommand{\sssp}{{\sc sssp}\xspace}
\newcommand{\diamp}{{\sc Diameter}\xspace}

\newcommand{\union}{\mathit{union}}

\newcommand{\vd}{\mbox{\sc vd}}
\newcommand{\dsource}{D_{\mathrm{src}}}

\newcommand{\stackpi}{\cS_{\pi}}
\newcommand{\stackpir}{\cS_{\pi}^*}
\newcommand{\lev}{\mathit{lev}}
\newcommand{\owner}{\mathit{owner}}

\newcommand{\myin}{\mathrm{in}}
\newcommand{\Finj}{F^{\myin}_j}
\newcommand{\Fxj}{F^{\times}_j}
\newcommand{\mystar}{\mathrm{star}}
\DeclareMathOperator{\rep}{rep}
\newcommand{\ec}{\operatorname{ecc}}
\newcommand{\pt}[1][v,Q(C)]{\ensuremath{P_{#1}} }

\newcommand{\home}[1]{\operatorname{home}(#1)}

\title{Single-Source Shortest Paths and Almost Exact Diameter in Pseudodisk Graphs} 
\newif
\ifShort

\author{Mark de Berg}{Department of Mathematics and Computer Science, TU Eindhoven, the Netherlands}{M.T.d.Berg@tue.nl}{https://orcid.org/0000-0001-5770-3784}{}
\author{Bart M.P.~Jansen}{Department of Mathematics and Computer Science, TU Eindhoven, the Netherlands}{B.M.P.Jansen@tue.nl}{https://orcid.org/0000-0001-8204-1268}{}
\author{Jeroen S.K.~Lamme}{Department of Mathematics and Computer Science, TU Eindhoven, the Netherlands}{J.S.K.Lamme@tue.nl}{https://orcid.org/0009-0005-8901-2271}{}

\authorrunning{M.~de Berg, B.M.P.~Jansen, and J.S.K.~Lamme}
\Copyright{Mark de Berg and Bart M.P.~Jansen and Jeroen S.K.~Lamme}
\funding{MdB and BJ are supported by the  Dutch Research Council (NWO) through Gravitation-grant NETWORKS-024.002.003.
         Part of the work by MdB was done while he was visiting IISc Bengaluru as Rukmini-Gopalakrishnachar Chair.}


\EventEditors{}
\EventNoEds{0}
\EventLongTitle{}
\EventShortTitle{}
\EventAcronym{}
\EventYear{}
\EventDate{}
\EventLocation{}
\EventLogo{}
\SeriesVolume{}
\ArticleNo{}

\begin{document}

\maketitle

\begin{abstract}
We study {\sc Single-Source Shortest-Path} (\sssp) on unweighted \emph{intersection graphs} whose
node set corresponds to a set $\D$ of $n$ constant-complexity objects in the plane. 
We prove \sssp can be solved in $O(U(n) \polylog n)$ expected time for any class 
of objects whose union complexity is~$U(n)$. In particular, we obtain an $O(n 2^{\alpha(n)} \log^2 n)$ 
algorithm for arbitrary pseudodisks,
and an $O(\lambda_{s+2}(n) \hspace{0.2mm}2^{O(\log^* n)} \log^2 n)$ algorithm for locally fat objects,
where $s$ is the maximum number of intersections between any pair of boundary arcs of the objects,
and $\lambda_t(n)$ is the the maximum length of a Davenport-Schinzel sequence of order $t$ 
on $n$ symbols. This significantly extends the class of objects
for which \sssp can be solved in $O(n \polylog n)$ time: so far, $O(n \polylog n)$
\sssp algorithms were not even known for pseudodisks that are fat \emph{and} convex \emph{and} similarly-sized.
\medskip

Our second result concerns the \diamp problem, which asks for the maximum distance
between any two nodes in a graph. Even for intersection graphs,
near-quadratic algorithms are already difficult to obtain, and
the $O(n^2 \polylog n)$ running time that follows from our \sssp algorithm
is the first near-quadratic running time for such general classes of intersection graphs.
Obtaining subquadratic running time is even more challenging. We prove that
the diameter of a set of arbitrary pseudodisks can be computed \emph{almost exactly},
namely up to an additive error of~2, in $\tilde{O}(n^{2-1/14})$ expected time.
This generalizes and speeds up a recent algorithm by Chang, Gao, and Le~(SoCG 2024) 
that works for similarly-sized disks (or similarly-sized pseudodisks
that are fat and satisfy a strong monotonicity assumption) and runs in 
$\tilde{O}(n^{2-1/18})$ time. 
To this end, we develop a so-called \emph{star-based $r$-clustering} for 
intersection graphs of pseudodisks---this is similar to an $r$-division for
planar graphs---which is interesting in its own right.
Our star-based $r$-clustering can also be used to
obtain an almost exact distance oracle for pseudodisks
that uses $O(n^{2-1/13})$ storage and has $O(1)$ query time.
\end{abstract}

\section{Introduction}
\label{sec:intro}

\subparagraph{The \sssp problem.}
{\sc Single-Source Shortest-Path} (\sssp) asks to compute
the distance from a given source node to all other nodes in a graph $\G=(V,E)$. 
In an unweighted graph, the problem can be solved in $O(n+m)$
time using breadth-first search, where $n := |V|$ and $m := |E|$.
\sssp is one of the most fundamental algorithmic graph problems---perhaps even \emph{the} most fundamental---and
it shows up as a subproblem in many applications.

We are interested in \sssp on (unweighted) \emph{intersection graphs}.
Let $\D$ be a set of $n$ objects in the plane. Then the intersection graph
$\ig[\D]$ has $\D$ as its node set, and it has an edge between two objects $D,D'\in \D$
iff $D$ intersects~$D'$. If the objects in $\D$ are (unit) disks, then $\ig[\D]$
is called a \emph{(unit-)disk graph}. Similarly, \emph{segment graphs} and
\emph{pseudodisk graphs} are intersection graphs where $\D$ is a set
of line segments or pseudodisks, respectively.
(A \emph{set of pseudodisks} is a set $\D$ of 
topological disks such that any two objects in~$\D$ are either 
disjoint or their boundaries intersect in at most two points~\cite{pseudodisk-union-86}.)
Since an intersection graph can have $\Omega(n^2)$ edges, solving \sssp by
computing its edges explicitly and then running a BFS takes quadratic time in the worst case.
The edges of an intersection graph are defined implicitly, however, so 
there is hope that \sssp on intersection graphs can be solved in subquadratic time.
This is indeed possible for many types of objects.

One approach to obtain a subquadratic \sssp algorithm is to implement BFS using a dynamic
(deletion-only) data structure for intersection queries
on the set~$\D$. Given an object $D\in \D$ at distance~$j$ from the source, we use
the data structure to find the not-yet visited objects that intersect~$D$
(which thus have distance $j+1$ from the source) and delete them from the data structure.
For line segments there exists such a data structure (a partition tree) that has
$\tilde{O}(n^{4/3})$ preprocessing time\footnote{We use the $\tilde{O}$-notation to suppress 
polylogarithmic factors.} and $\tilde{O}(n^{1/3})$ query and update time. 
Thus, \sssp on segment intersection graphs can be solved in $\tilde{O}(n^{4/3})$ time.
This is likely optimal, since Hopcroft's problem---given a set~$L$ of $n$ lines
and a set~$P$ of $n$ points, decide if there is a point-line incidence---can be solved
by running \sssp on $\ig[L\cup P]$ with one of the lines as the source.
A running time of $O(n^{4/3})$ is conjectured to be optimal for Hopcroft's problem,
and Erickson~\cite{DBLP:journals/dcg/Erickson96} actually proved an $\Omega(n^{4/3})$ lower
bound for so-called partitioning algorithms.

For other types of objects, however, it is possible to achieve $O(n\polylog n)$ running times.
For axis-aligned segments, for example, the approach sketched above immediately
gives an $O(n\polylog n)$ \sssp algorithm using standard data structures for orthogonal 
segment-intersection queries. $O(n\polylog n)$ algorithms also exist for disk graphs,
on which much of the research on \sssp in intersection graphs has focused. 
For unit disks, Cabello and Jej{\v c}i{\v c}~\cite{CabelloJ15} and Chan and Skrepetos~\cite{ChanS16}
gave $O(n \log n)$ algorithms, which are optimal in the algebraic decision-tree model. 
The semi-dynamic data structure of Efrat, Itai, and Katz~\cite{EfratIK01} 
also gives an $O(n \log n)$ algorithm.
For arbitrarily-sized disks, Kaplan~\etal~\cite{KaplanMRSS20} and Liu~\cite{doi:10.1137/20M1388371}
presented the first $O(n\polylog n)$ algorithms, obtaining a running time 
of~$O(n\log^4 n)$.  Klost~\cite{Klost23} recently improved this to~$O(n\log^2 n)$.
Even more recently, Brewer and Wang~\cite{DBLP:conf/esa/Brewer025} and
De~Berg and Cabello~\cite{DBLP:conf/esa/BergC25} independently presented
optimal $O(n\log n)$ \sssp algorithms on disk graphs.
The algorithms of Klost and of De~Berg and Cabello both rely on an
efficient intersection-detection data structure for disks, which can be
obtained from an additively weighted Voronoi diagram on the disk centers.
The algorithm of Brewer and Wang also relies on additively weighted
Voronoi diagrams and Delaunay triangulations. 

For pseudodisks, approaches
based on (weighted) Voronoi diagrams will not work---for that, the objects need to
be the same up to scaling and translation. Moreover, data structures with
$O(n\polylog n)$ preprocessing and $O(\polylog n)$ query time for intersection
queries do not exist for pseudodisks. Chang, Gao, and Le~\cite{DBLP:journals/corr/abs-2401-12881}
recently observed that the \sssp algorithm of Chan and Skrepetos~\cite{ChanS19} for unit disks,
can be adapted to pseudodisks that are fat and similarly-sized and have the following strong monotonicity property:
there is a point~$c$ inside each object~$D$ such that the parts 
of~$D$ on either side of the horizontal line through $c$ are $x$-monotone and
the parts of~$D$ on either side of the vertical line through $c$ 
are $y$-monotone.\footnote{This condition is not stated in~\cite{DBLP:journals/corr/abs-2401-12881},
but it is needed for the correctness (J.~Gao, personal communication).}
Observe that convexity of the pseudodisks does not guarantee this property.
Thus, the following question is still wide open: can \sssp on pseudodisk graphs be solved in 
$O(n \polylog n)$ time? More generally, which objects admit $O(n\polylog n)$
\sssp algorithms? Do they need to be fat (except for the simple case of axis-aligned segments)?
Or do they need to be similarly-sized, except when they admit efficient 
intersection-detection data structures? 
We answer these questions: all that is needed is that the
union complexity of the objects is near-linear.

\subparagraph{Our results and techniques on \sssp.}
We define an \emph{object} to be a topological disk whose boundary consists of 
constantly many $x$-monotone arcs, each of which is described by a polynomial
of constant maximum degree (plus its endpoints).
The assumption that the boundary arcs are $x$-monotone is without loss of generality,
since a constant-complexity arc can always be split into $O(1)$ $x$-monotone pieces.
We denote the maximum number of intersection points between any two boundary 
arcs by~$s$; because the boundary arcs have constant complexity, $s$~is also a constant. 

Our \sssp algorithm is efficient when the union complexity of the input set $\D$ is small:
its expected running time is $O(U(n)\polylog n)$, where $U(n)$ denotes the union 
complexity of the objects in~$\D$. More precisely, $U(\cdot)$ is a function such that
the combinatorial complexity of the union of any subset $\D'\subseteq \D$ is at most~$U(|\D'|)$.
It is well known that $U(n)=O(n)$ for pseudodisks~\cite{pseudodisk-union-86}, and that
$U(n)=O(n\hspace{0.2mm} 2^{O(\log^* n)})$ for so-called locally 
fat\footnote{We omit the somewhat technical definition of locally fat objects~\cite{DBLP:journals/dcg/Berg08};
intuitively an object is locally fat if it does not contain any long and thin parts.} 
objects~\cite{AronovBES14}. Thus, for these classes of objects
we obtain an $O(n \polylog n)$ algorithm. 

The following theorem states our result more precisely. The function $\lambda_{t}(n)$ 
in the theorem is the maximum length of a Davenport-Schinzel sequence of order $t$ 
on $n$ symbols~\cite{DBLP:books/daglib/0080837}, which is near-linear for any fixed constant~$t$.
In particular, $\lambda_{t}(n)= O(n\log^* n)$ for any constant~$t$~\cite{Szemeredi-DS};
more precise bounds, which depend on the exact value of~$t$, are known as well~\cite{DBLP:journals/jacm/Pettie15}.
\begin{restatable}{theorem}{maintheorem} \label{thm:sssp}
Let $\D$ be a set of $n$ constant-complexity objects from a family of objects with 
union complexity~$U(n)$, and such that any two boundary arcs intersect at most~$s$ times.
Then we can solve \sssp on the intersection graph~$\ig[\D]$ in 
$O(\lambda_{s+2}(U(n))\log^2 n)$
expected time. In particular, the expected running time is $O(\lambda_{s+2}(n)\log^2 n)$ for
pseudodisks and it is $O(\lambda_{s+2}(n)2^{O(\log^* n)}\log^2 n )$ for locally fat objects.
\end{restatable}
At a very high level, our algorithm follows the same approach as used by Klost~\cite{Klost23}
and by De~Berg and Cabello~\cite{DBLP:conf/esa/BergC25}. Let $\dsource$ be the source object,
define $L_j\subset \D$ to be the set of objects at distance~$j$ from~$\dsource$, and define
$\union(L_j)$ to be the union of the objects in~$L_j$.
Then, instead of querying with each object~$D\in L_j$ to discover the objects in~$L_{j+1}$, 
we do the following: we identify a set $\C_{j+1}\supseteq L_{j+1}$ of \emph{candidate objects},
and we determine which of them intersect $\union(L_j)$ (and are thus in~$L_{j+1}$). 
The main challenges are 
(i)~to ensure that each object is a candidate only a few times, and
(ii)~to determine efficiently which of the candidates in $\C_{j+1}$ intersect~$\union(L_j)$.
Klost~\cite{Klost23}, and De~Berg and Cabello~\cite{DBLP:conf/esa/BergC25}, heavily use the fatness of the objects
to overcome~(i), and they use Voronoi diagrams (or, for fat triangles, intersection-detection data structures)
to overcome~(ii). Hence, we must develop completely different tools. 

Our main tool, which is useful to overcome both challenges, is what we call a \emph{random stacking}.
This is a subdivision of $\union(\D)$ obtained by adding the objects ``on top of each other''
in random order; see \autoref{fig:stacking}. The important properties of a random stacking
are that the expected number of faces of the subdivision is~$O(U(n))$, that the expected
number of crossings between object boundaries and face boundaries is~$O(U(n)\log n)$,
and that each face~$f$ is fully contained in at least one of the objects. 
Our \sssp algorithm will be guided by the stacking: to compute $L_{j+1}$ we determine
the faces of the stacking that intersect $\union(L_{j})$, we collect the objects
whose boundaries intersect the faces---these form the candidate set~$\C_{j+1}$---and 
then we determine which candidates intersect~$\union(L_{j})$.
The fact that each face is fully contained in at least one object helps us
to guarantee property~(i) above. Turning these ideas into an $O(n\polylog n)$ algorithm
is not easy, however, and requires several other ideas. For example,
we do not know how to compute the intersections between the object boundaries and the face boundaries
of a random stacking fast enough. Hence, we have to simplify the stacking
so that its faces do not have holes. Second, determining which candidates
intersect~$\union(L_{j})$ is difficult, because $\union(L_{j})$ need not be connected and so
$\Reals^2\setminus \union(L_{j})$ can have holes.
We therefore do this separately for each face~$f$, which helps
because we can prove that each component of $f\setminus \union(L_{j})$ is simply connected. 

\subparagraph{The \diamp problem.}
The \diamp problem, which asks for the maximum distance between any two nodes in a graph,
is another fundamental algorithmic problem on graphs. In general graphs, it can be solved
in $O(nm)$ time by running \sssp $n$ times, once with each node as the source. One can also use fast matrix multiplication
to obtain an $O(n^{\omega})$ time algorithm, where $\omega \approx 2.37$ is 
the matrix-multiplication exponent, which is faster for
dense graphs. A major open problem has been whether there is a truly subquadratic algorithm for 
sparse graphs. Cabello~\cite{DBLP:journals/talg/Cabello19} achieved a breakthrough in this area 
by showing that \diamp on planar graphs can be solved in $\tilde{O}(n^{11/6})$ expected time.
(His algorithm even works for weighted graphs.) This was improved to $\tilde{O}(n^{5/3})$
by Gawrychowski~\etal~\cite{DBLP:journals/siamcomp/GawrychowskiKMS21}.

For intersection graphs the situation is less rosy: even deciding
whether the diameter of  an intersection graph of unit-length segments
is at most~3 cannot be done in $O(n^{2-\eps})$ time for any $\eps>0$, 
assuming the Strong Exponential-Time Hypothesis~\cite{DBLP:conf/compgeom/BringmannKKNP22}.
The same holds for unit-size equilateral triangles.
For unit disks, Chan~\etal~\cite{DBLP:conf/focs/ChanCGKLZ25} recently presented
an $O(n^{2-1/18+o(1)})$ algorithm for \diamp. (For arbitrarily-sized axis-aligned squares
and axis-aligned unit squares, they obtain $\tilde{O}(n^{2-1/12})$ and $O(n^{2-1/8+o(1)})$ 
running time, respectively.) It is still open whether
a subquadratic algorithm for \diamp exists for arbitrarily-sized disks.
In fact, even near-quadratic algorithms are rare: so far they only exist
for classes of objects for which \sssp admits a near-linear algorithm.
Thus, our \sssp algorithm for objects with near-linear union size provides
the first near-quadratic exact algorithm for \sssp for pseudodisks
and for general fat objects.

The difficulty of \diamp led Chang, Gao and Le~\cite{Chang0024}
to study the problem of computing the diameter \emph{almost exactly}, namely,
up to an additive constant. They presented a subquadratic algorithm for 
similarly-sized disks (and for similarly-sized
fat pseudodisks satisfying a strong monotonicity condition; see above). Their algorithm runs in~$\tilde{O}(n^{2-1/18})$ time
and it computes the diameter up to an additive error of~2. In the arXiv version
of the paper~\cite{DBLP:journals/corr/abs-2401-12881}, they improve the error to~1.
To obtain this result, they first prove that the so-called \emph{distance VC-dimension} 
and the \emph{distance-encoding VC-dimension}
of a pseudodisk graph are at most~4; see~\cite[Theorems 12 and 14]{Chang0024}. Their proofs are purely topological,
so they do not require the pseudodisks to be fat, convex, or similarly sized.
They then describe a framework to compute the diameter almost exactly,
which requires two ingredients: an algorithm that computes a so-called 
\emph{clique-based $r$-clustering} for the given graph class---this is a generalization
of \emph{$r$-divisions}~\cite{DBLP:journals/siamcomp/Frederickson87} for planar graphs---and an $O(n\polylog n)$ \sssp algorithm.
They present an algorithm that computes a clique-based clustering
for similarly-sized fat pseudodisks, and an \sssp algorithm for similarly-sized fat
objects satisfying a monotonicity condition.

\subparagraph{Our results on \diamp.}
To obtain a \diamp algorithm for pseudodisks, we need to supply two
ingredients: a suitable $r$-clustering and an \sssp algorithm. The latter is the
first contribution of our paper, as discussed above. For the former we observe
that the subgraphs that make up the boundary in the $r$-clustering need not be
cliques, as was the case in~\cite{Chang0024}: it 
suffices that each of these subgraphs has a representative node that is a neighbor of
any other node in the subgraph. Thus, instead of working with a clique-based 
$r$-clustering, we may also work with a \emph{star-based $r$-clustering}. 
We show how to compute a star-based $r$-clustering for a
pseudodisk graph---in fact, our algorithm works for any class of objects with
near-linear union complexity---with similar quality guarantees as the clique-based
$r$-clustering of Chang, Gao and Le~\cite{Chang0024}; interestingly, 
our algorithm for computing an $r$-clustering also uses random stackings. 
Since $r$-divisions for planar graphs have found many applications,
we expect that our star-based $r$-clusterings may find other applications as well.

Our star-based $r$-clustering has a stronger bound on the number
of stars that comprise the boundary of a cluster, as compared to the
number of cliques that comprise the boundary of a cluster in the clique-based
clustering of~\cite{Chang0024}. This allows us to not only obtain an
algorithm for a much larger class of objects, but it also allows us 
to improve the $\tilde{O}(n^{2-1/18})$ running time,
as stated in the following theorem.
\begin{restatable}{theorem}{diamtheorem} \label{thm:diam}
Let $\D$ be a set of $n$ constant-complexity pseudodisks.
Then we can solve \diamp on the intersection graph~$\ig[\D]$ in $\tilde{O}(n^{2-1/14})$ 
expected time, up to an additive error of~2. 
\end{restatable}

\subparagraph{Distance oracles.}
A distance oracle for a graph $G=(V,E)$ is a data structure that can 
report the distance between any two query nodes~$u,v \in V$. 
For arbitrary graphs, exact distance oracles require $\Omega(n^2)$ bits of storage in the worst case,
irrespective of query time~\cite{DBLP:journals/jacm/ThorupZ05}.
Hence, research on distance oracles focused on special graph classes such as planar graphs, 
where the goal is to obtain sublinear query time with subquadratic storage. 
The survey by Sommer~\cite{DBLP:journals/csur/Sommer14} gives an overview of the work done until 2014,
and the paper by Charalampopoulos~\etal~\cite{DBLP:journals/jacm/CharalampopoulosGLMPWW23} and the references therein 
give an overview of more recent work.

For geometric intersection graphs only a few results are known~\cite{DBLP:journals/algorithmica/AronovBT25,DBLP:conf/focs/ChanCGKLZ25,chan-skrepetos,DBLP:journals/corr/abs-2401-12881,DBLP:conf/isaac/BergJL25,gao-zhang}. 
Both Gao and Zhang \cite{gao-zhang} and Chan and Skrepetos \cite{chan-skrepetos} show that an $(1+\eps)$-approximate distance oracle 
with $O(n\log n)$ storage\footnote{Here (and in the results discussed later) the storage is
not measured in bits, but in the real RAM model.}
and $O(1)$ query time exists for unit-disk graphs;
their oracle also works in the weighted setting.
Chang, Gao and Le~\cite{DBLP:journals/corr/abs-2401-12881} show that their technique
for computing the diameter almost exactly, can also be used to obtain a distance oracle
for similarly-sized fat pseudodisks
that uses $O(n^{2-1/18})$ storage, has query time~$O(1)$, and an additive error of~1.
The paper by Chan~\etal~\cite{DBLP:conf/focs/ChanCGKLZ25} that provides the first
exact subquadratic \diamp algorithm  for unit disks also provides the first
exact distance oracles for unit disks; it uses $O(n^{2-1/20+o(1)})$ storage
and has $O(\log n)$ query time. (For axis-aligned squares they obtain better bounds.)
We show that our star-based $r$-clustering can be used to extend the
distance oracle of Chang, Gao and Le to arbitrary pseudodisks, and improve the storage bound
at the same time.
Our oracle uses $O(n^{2-1/13})$ storage---note that this bound is smaller than
the computation time of our \diamp algorithm---and has $O(1)$ query time,
and it can report distances with an additive error to~2.

\subparagraph{Preliminaries.}
To focus on the algorithmic and combinatorial aspects of the problem,
and not on algebraic issues, we work in the real-RAM model of computation and
assume that all basic operations on boundary arcs---computing
the intersection points between two arcs, testing if a point lies above
an arc, et cetera---can be carried out in~$O(1)$ time. 
This is common practice in computational geometry. 

We denote the boundary and the interior of a compact set~$R$---this can be an object, 
a union of objects, or a face in a subdivision---by $\bd R$ and  $\intr(R)$, respectively. 
For a subdivision $\cS$, we use $|\cS|$ 
to denote its (combinatorial) complexity, that is, its total number of vertices,
edges (which can be constant-complexity arcs), and faces. Similarly, 
$|f|$ denotes the complexity of a face~$f$. With a slight abuse of nation,
we sometimes write $f\in\cS$ when $f$ is a face of~$\cS$.

To simplify the presentation, we assume that the input set $\D$ of objects is in 
general position.
In particular, no three object boundaries pass through
a common point and all intersections between object boundaries are proper intersections.
(The latter assumption means that for any two distinct objects $D,D'\in \D$, the
set $\bd D \cap \bd D'$ consists of finitely many points and that there are no tangencies
between object boundaries.)

\section{Single-source shortest paths}
\label{sec:sssp}
Before we present our \sssp algorithm, we first develop the main technical tool on which
the algorithm is based: random stackings.

\subsection{Random stackings}

\subparagraph{Stackings.}
Let $\pi$ be a permutation of~$\{1,\ldots,n\}$, which defines an ordering of the objects 
in our input set~$\D$.
(See the introduction for the definition of ``object''.)
Imagine drawing the objects one by one in the order given by~$\pi$, such that each object 
is drawn on top of the previous ones, that is, each object
hides the parts of the objects already drawn that are covered by it. We call the resulting
subdivision of $\union(\D)$ a 
\emph{stacking}\footnote{Note that 
$\stackpi(\D)$ is a subdivision of $\union(\D)$, not of~$\Reals^2$.}  
of the set~$\D$,
and we denote it by~$\stackpi(\D)$; see \autoref{fig:stacking}(i).
\begin{figure}
    \centering
    \includegraphics{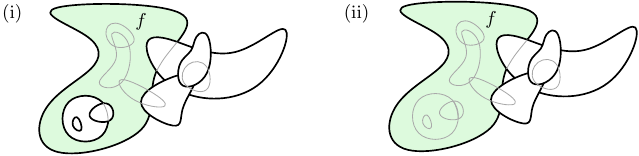}
    \caption{(i) A stacking $\stackpi(\D)$ of eleven objects. The face boundaries are shown thick and in black;
                 the parts of the object boundaries that are hidden are shown thin and in gray.
    (ii) The simplified stacking; the green face $f$ has absorbed the faces
                 that were fully contained in its hole.}
    \label{fig:stacking}
\end{figure}
More formally, let $D_1,\ldots,D_n$ be the ordering on $\D$ as specified by~$\pi$, 
and define $\D_{>i} := \{ D_{i+1},\ldots, D_n \}$. Then the interiors of the faces of $\stackpi(\D)$
are the connected components of the sets $\intr(D_{i}) \setminus \union(\D_{>i})$,
over all $1\leq i\leq n$. Similarly, the edges of $\stackpi(\D)$ contributed by $D_i$
are the connected components of the sets 
$\{ e \setminus \union(\D_{>i}) : \mbox{ $e$ is a boundary arc of $D_i$} \}$,
and the vertices of $\stackpi(\D)$ are the points in the 
sets $\bd D_{i} \cap \bd(\union(\D_{>i} ))$ plus the vertices
of~$D_i$ that are not hidden by~$\union(\D_{>i})$, over all $1\leq i\leq n$.
A straightforward but important property of $\stackpi(\D)$ is the following:
for each face $f$ in $\stackpi(\D)$ there is an object~$D_i\in \D$ such that $f\subset D_i$.
We call this object the  \emph{owner} of~$f$ and we denote it by~$\owner(f)$.
If there are multiple such objects, we can assign any one of them to be the owner. 
(Our algorithm will pick $\owner(f)$ as the last object in the order specified by~$\pi$
that contains~$f$.) Note that an object can own multiple faces.
\medskip

To be able to work efficiently with~$\stackpi(\D)$, 
we want its complexity $|\stackpi(\D)|$ to be small. In the worst case, $|\stackpi(\D)|=\Theta(n^2)$, even when 
$\D$ is a set of disks, but when we generate $\stackpi(\D)$ from a random 
permutation---we call this a \emph{random stacking}---then $|\stackpi(\D)|$ is expected to be small
when the objects in~$\D$ have small union complexity.
\begin{lemma}\label{lem:size-of-stacking}
Let $\D$ be a set of $n$ objects from a family of objects with union complexity~$U(n)$.
Let $\pi$ be a random permutation. Then $\Ex{|\stackpi(\D)|}=O(U(n))$.
\end{lemma}
\begin{proof}
Let $V$ be the set of vertices of $\stackpi(\D)$ that are an intersection
between two object boundaries. The other vertices of $\stackpi(\D)$ are object vertices, 
of which there are~$O(n)$ since the objects have constant complexity.
Since $|\stackpi(\D)|=O(|V|+n)$, it suffices to bound $\Ex{|V|}$.
Let $D_1,\ldots,D_n$ be the ordering of the objects in~$\D$, as specified by~$\pi$.
Let $P:= \{ \bd D_i \cap \bd D_j : 1\leq i<j\leq n\}$ be the set of all intersection points of the 
boundaries of the objects in~$\D$. (Because the object boundaries have
constant complexity, and we assume general position, $|P|=O(n^2)$.)
Clearly, $V\subseteq P$. Let $\D(p) := \{ D_i\in \D : p\in \intr(D_i)\}$ denote the objects
containing the point $p \in P$ in their interior, and define $\lev(p) := |\D(p)|$ to be the level of~$p$.
Then a point $p\in \bd D_i \cap \bd D_j$ shows up as a vertex of $\stackpi(\D)$ iff
both $D_i$ and $D_j$ come after any object $D_k\in \D(p)$ in the order specified by~$\pi$.
This implies that 
\[
\Pr{p\in V} = \tfrac{2}{(\lev(p)+2)(\lev(p)+1)}.
\]
Hence, if we define $P_{\ell} := \{ p\in P : \lev(p) =\ell\}$, then
\[
\Ex{|V|} = \sum_{p\in P} \Pr{p\in V} 
         = \sum_{p\in P} \tfrac{2}{(\lev(p)+2) (\lev(p)+1)}
         = \sum_{\ell=0}^{n-2}|P_{\ell}| \cdot \tfrac{2}{(\ell+2) (\ell+1)}.
\]
Observe that for a given number of intersection points, this expression 
is maximized when the levels of the point in $P$ are as low as possible.
Moreover, for any $0\leq \ell^*\leq n-2$ we have
\[
\sum_{0\leq \ell\leq \ell^*} |P_{\ell}| = O(U(n) \cdot (\ell^*+1))
\]
by the Clarkson-Shor technique~\cite{DBLP:journals/dcg/ClarksonS89}; 
see also \cite[Theorem~1.1]{DBLP:journals/cpc/Sharir03}.
This implies that to bound  $\Ex{|V|}$, we may assume that $|P_{\ell}|=\Theta(U(n))$ for all~$\ell$.
This gives
\[
\Ex{|V|} = \sum_{\ell=0}^{n-2}|P_{\ell}| \cdot \tfrac{2}{(\ell+2) (\ell+1)}
         =  O\left( U(n) \cdot\sum_{\ell=0}^{n-2}\tfrac{2}{(\ell+2) (\ell+1)} \right) = O(U(n)),
\]
which finishes the proof.
\end{proof}

\subparagraph{Conflict lists.}
Every face~$f$ of $\stackpi(\D)$ is fully contained in at least one object from~$\D$, but
there can also be objects in $\D$ whose boundary intersects~$f$. We say that such objects
are \emph{in conflict with} $f$ and we define
\[
K_f := \{ D\in \D : \bd D \cap f \neq \emptyset \}
\]
to be the set of objects in conflict with~$f$.
Note that $K_f$ not only contains the objects whose boundary intersects $\intr(f)$,
but (perhaps counter-intuitively) also those that contribute an edge to~$\bd f$.
Conflict lists\footnote{The term \emph{conflict list}
is also used in the analysis of randomized incremental construction (RIC) algorithms. 
However, our stackings depend on the order in which the objects are added and our faces
have non-constant complexity, so we cannot use standard results on RIC.} 
play an important role in our \sssp algorithm, and we need their total size to be small. 
We also need a bound on the number
of proper intersections\footnote{A proper intersection
between an object boundary $\bd D_i$ and a face boundary~$\bd f$,
is a point $p\in \bd D_i\cap \bd f$ such that $\bd D_i$ enters the interior of~$f$ at point~$p$.} 
between object boundaries and face boundaries at various points in our analysis.
Fortunately, the next lemma shows that the number of such intersections---and, 
hence, the total size of the conflict lists---is expected to be small.
\begin{lemma} \label{lem:size-of-conflict-sets}
Let $\D$ be a set of $n$ objects from a family of objects with union complexity~$U(n)$.
Let $\pi$ be a random permutation and let $F$ be the set of faces of $\stackpi(\D)$.
For a face $f\in F$, define
$
K^*_f  :=  \{ p:  \mbox{$p$ is a proper intersection between $\bd D_i$ and $\bd f$ for some $D_i\in D$}   \}.
$
Then $\Ex{\sum_{f\in F} |K^*_f|}=O(U(n)\log n)$ and $\Ex{\sum_{f\in F} |K_f|}=O(U(n)\log n)$.
\end{lemma}
\begin{proof}
Consider a face $f\in F$. Any object $D_i\in K_f$ is either completely contained in $f$
or its boundary~$\bd D_i$ intersects~$\bd f$. Trivially, for any object $D$ there is at most one face $f$ 
of the former type. Moreover, the total total number of edges of the
faces in $F$ is $O(U(n))$, so it remains to count the number of
objects whose boundary properly intersects a face boundary.
Thus, if we can show that $\Ex{\sum_{f\in F} |K^*_f|}=O(U(n)\log n)$ then we also
show that $\Ex{\sum_{f\in F} |K_f|}=O(U(n)\log n)$.
The proof that $\Ex{\sum_{f\in F} |K^*_f|}=O(U(n)\log n)$ is similar to the proof of \autoref{lem:size-of-stacking}, 
as explained next.

Let $P$ be the set of all intersection points between the object boundaries, 
let $\D(p) := \{ D_i\in \D : p\in \intr(D_i)\}$ denote the objects containing a point~$p\in P$ 
in their interior, let $\lev(p) := |\D(p)|$ be the level of~$p$, and let
$P_{\ell} := \{ p\in P : \lev(p) =\ell\}$. Define $K^* := \bigcup_f K^*_f$
and consider a point $p\in K^*$ contributed by a pair $(D_i, f)$. Let $D_j$ be 
the unique object
such that $p\in \bd D_i \cap \bd D_j$. Thus, $D_j$ contributes an
edge to $\bd f$ that is intersected by~$\bd D_i$. For this to happen,
$D_j$ must come after any object $D_k\in \D(p)$ in the order specified by~$\pi$.
Hence, for a point $p\in P$ we have
\[
\Pr{ p\in K^*} \leq \tfrac{1}{\lev(p)+1}.
\]
Following the same reasoning as in the proof of \autoref{lem:size-of-stacking},
we obtain
\[
\Ex{|K^*|} = \sum_{\ell=0}^{n-2}|P_{\ell}| \cdot \tfrac{1}{\ell+1}
         =  O\left( U(n) \cdot\sum_{\ell=0}^{n-2}\tfrac{1}{\ell+1} \right) = O(U(n)\log n),
\]
which finishes the proof.
\end{proof}

\subparagraph{Computing a stacking.}
The stacking $\stackpi(\D)$ can be computed by a divide-and-conquer algorithm.
To this end, let $D_1,\ldots,D_n$ be the ordering of the objects as specified by~$\pi$.
We recursively compute a stacking~$\cS_1$ for the ordered set $\D_1 := \{D_1,\ldots,D_{\floor{n/2}} \}$
and a stacking~$\cS_2$ for the ordered set $\D_2 := \{D_{\floor{n/2}+1},\ldots, D_n\}$.
Next, we compute the overlay of $\cS_1$ and $\cS_2$, where we label 
each face in the overlay with the face(s) of $\cS_1$ and/or $\cS_2$ containing it. 
This can be done in $O((n_1+n_2)\log (n_1+n_2)+k)$ expected time by Mulmuley's 
algorithm~\cite{DBLP:journals/jacm/Mulmuley91},
where $n_1$ and $n_2$ are the complexity of $\cS_1$ and $\cS_2$, respectively,
and $k$ is the number of intersections between the edges of $\cS_1$ and the edges of $\cS_2$. 
Observe that we can easily compute owners for the new faces:
faces labeled by a face from $\D_1$ inherit the owner from that face, 
and all other faces inherit their owner from the corresponding face in~$\D_2$.
To finish the construction of~$\stackpi(\D)$,
it remains to merge adjacent faces of the overlay that are contained
in the same face of~$\cS_2$. (The objects of $\D_2$  hide those of~$\D_1$,
so we must undo the partitioning of the faces in $\cS_2$ caused by the edges of~$\cS_1$.)
We obtain the following lemma.
\begin{lemma}\label{lem:compute-stacking}
Let $\D$ be a set of $n$ objects from a family of objects with union complexity~$U(n)$.
Then we can compute in $O(U(n)\log^2 n)$ expected time a stacking $\stackpi(\D)$, 
including an owner for each face of $\stackpi(\D)$, such that $|\stackpi(\D)|=O(U(n))$.
\end{lemma}
\begin{proof}
Consider the algorithm described above. Because $\D$ is randomly ordered, 
we can assume in our analysis that $\D_1$ and $\D_2$ are randomly ordered as well.
Thus, $\Ex{|\cS_1|} = O(U(n))$ and $\Ex{|\cS_2|}= O(U(n))$. 
Now observe that if an edge~$e_1$ of $\cS_1$ intersects an edge $e_2$ of $\cS_2$,
then $e_1\cap e_2$ is a point of one of the sets $K^*_f$ defined in \autoref{lem:size-of-conflict-sets}.
Hence, $\Ex{k} = O(U(n)\log n)$, where $k$ denotes the number of intersections between 
the edges of $\cS_1$ and the edges of $\cS_2$. The expected time needed for the merge step
is therefore $O(U(n)\log n)$. Since the recursion depth is $O(\log n)$ and each object 
is handled once at each level of the recursion, by linearity of expectation 
the total expected running time of our algorithm is $O(U(n)\log^2 n)$.

The algorithm as described so far computes a stacking whose \emph{expected}
complexity is~$O(U(n))$. This can be turned into an algorithm that is guaranteed
to compute a stacking with the desired complexity in a standard manner:
simply repeat the algorithm until a stacking of at most twice the expected complexity
is reached. This increases the expected running time by a factor~2.
\end{proof}

\subparagraph{Simplified stackings.}
Computing a stacking can be done efficiently, as we just saw, but
computing the conflict lists is more challenging. We would like to do this 
by tracing the boundary $\bd D_i$ of each object $D_i\in\D$ through~$\stackpi(\D)$, and add $D_i$ to the
conflict list of each face~$f$ that we encounter. The problem is that a face~$f$ may
contain many holes, and then it is hard to figure out which holes are being crossed by~$\bd D_i$.
Indeed, given a set of $n$ point-size holes and a set of $n$ lines,
it seems impossible to determine in $o(n^{4/3})$ time which holes are crossed 
by which lines, because this would solve Hopcroft's 
problem.\footnote{This is not a formal lower bound, because the problem 
that we have to solve is not for an arbitrary input, but for an input 
arising from a random stacking. We have not been able to exploit this, however.}
We will therefore work with a simplified subdivision, for which it will be 
easier to compute its conflict lists. 

The new subdivision, which we call a \emph{simplified stacking}, 
is obtained by removing faces that are 
contained in a hole of some other face, and then adding these faces
to the outermost face containing them; see \autoref{fig:stacking}(ii) for an illustration.
More precisely, the new subdivision of $\union(\D)$, 
which we denote by $\stackpir(\D)$, is defined as follows: 
For each face $f$ of $\stackpi(\D)$ that does not lie inside a hole of some other face~$g$ of $\stackpi$,
there is a face $f^*$ in $\stackpir(\D)$ that is the union of $f$ and all the faces
of $\stackpi(\D)$ that are enclosed by~$f$.
\begin{lemma} \label{obs:simplify}
The simplified stacking $\stackpir(\D)$ of a stacking $\stackpi(\D)$ is a subdivision
of $\union(\D)$ such that each face $f^*$ of $\stackpir(\D)$ has the following properties:
\begin{enumerate}[(i)]
\item $f^*$ is the union of one or more faces of $\stackpi(\D)$, 
\item $f^*$ is simply connected, 
\item $f^*$ has an owner, that is, there is an object $D\in \D$ such that $f^*\subseteq D$.
\end{enumerate}
Furthermore, $\stackpir(\D)$ can be computed from $\stackpi(\D)$ in $O(|\stackpi(\D)|)$ time.
\end{lemma}
\begin{proof}
Property~(i) is immediate from the construction. To prove property~(ii),
consider a face~$f^*$ in $\stackpir(\D)$ and its corresponding face $f$ in $\stackpi(\D)$. 
The holes in~$f$ (if any) are contained in $\owner(f)$---and,
hence, in $\union(\D)$---because $\owner(f)$ is simply connected. 
This implies that these holes are fully covered by other faces in~$\stackpi(\D)$, which
become part of~$f^*$, and so $f^*$ no longer has holes.
Moreover, $f^* \subseteq \owner(f)$, which proves property~(iii).

Computing $\stackpir(\D)$ from $\stackpi(\D)$ in $O(|\stackpir(\D)|)$ time is trivial.
Indeed, the standard representation of planar subdivisions is a doubly-connected edge 
list~\cite{bcko-cgaa-08}, which contains for each face $f$ a list of pointers
to (an edge of) each of its holes. Thus, we can easily traverse $\stackpi(\D)$
and remove all information regarding faces that are no longer present in~$\stackpir(\D)$.
\end{proof}

\subparagraph{Computing the conflict lists of a simplified stacking.}
Because the faces of $\stackpir(\D)$ are simply connected, we can compute
their conflict sets efficiently, as explained next. Note that we can easily compute
for each boundary arc of an object $D_i\in \D$ in which face of~$\stackpir(\D)$ its left endpoint
is contained, using point location.
 (In fact, the algorithm we present below will
 compute this as well, so there is no need to do it separately.)
We also know which boundary arcs contribute to which face boundary.
Thus, to compute the conflict lists of the faces in $\stackpir(\D)$ it remains to
compute the intersections between the face boundaries and the boundary arcs of the objects.
This boils down to solving the 
following variant of the \emph{red-blue intersection problem}: given a set $\R$ of $x$-monotone
red arcs that are the edges of a 
connected\footnote{We assume for simplicity that $\union(\D)$
is connected, so that $\stackpir(\D)$ is also connected. If not, we simply run the
algorithm separately on each connected component of~$\union(\D)$.}
planar subdivision, and a set~$\B$ of 
$x$-monotone blue arcs, compute all intersections between a red arc and a blue arc.
We call these intersections \emph{purple}. Note that there are no intersections between
two red arcs (except at shared endpoints, which we do not consider to be intersections). 
There can be many intersections among the blue arcs---we call them \emph{blue intersections}---and
our goal is to avoid spending time on them. 
\medskip

Let $m := |\R|+|\B|$ and let $k$ be the number of purple intersections.
The general red-blue intersection problem requires $\Omega(m^{4/3}+k)$ time,
even when $\R$ is a connected set of line segments and $\B$ is a set of points. 
There are, however, variants that can be solved in $O((m+k)\polylog m)$ time.
In particular, 
Basch, Guibas, and Ramkumar~\cite{DBLP:journals/algorithmica/BaschGR03} 
showed that the version where the red set is connected and the blue set is 
connected---more formally, the union of the arcs in the red (resp.~blue) set should be connected---can 
be solved in $O(\lambda_{s+2}(m+k)\log^3 m)$ time. 
The running time was later improved to $O(\lambda_{s+2}(m+k)\log m)$ expected time
by Har-Peled and Sharir~\cite{DBLP:journals/dcg/Har-PeledS01}.
In our setting, $\B$ is not necessarily connected, but because
$\R$ is the edge set of a connected subdivision, we can still use the
technique of Har-Peled and Sharir. The key insight is that Har-Peled
and Sharir need $\B$ to be connected to be able to process $\R$ efficiently, but
processing $\R$ is easy anyway when
the red arcs do not cross each other.
\begin{theorem}\label{thm:red-blue}
Let $\R$ be a set of $x$-monotone red arcs that are the edges of a connected planar subdivision, 
and let~$\B$ be a set of $x$-monotone blue arcs. 
Let $s$ be the maximum number of intersections between any two arcs in $\R\cup\B$. 
Then we can solve the red-blue intersection problem on $\R\cup \B$ in $O(\lambda_{s+2}(m+k) \log m)$
expected time, where $m := |\R|+|\B|$ and $k$ is the number of purple intersections.
\end{theorem}
\begin{proof}
We first sketch the algorithm of Har-Peled and Sharir~\cite{DBLP:journals/dcg/Har-PeledS01},
which was designed for the setting where both $\R$ and $\B$ are connected sets, and then we 
adapt it to our setting.
\medskip

The algorithm sweeps a vertical line~$\ell$ from left to right over the plane, with the goal of
maintaining \emph{visibility pairs}, which are pairs consisting of a red arc~$\rho$
and a blue arc~$\beta$ that see each other along $\ell$. (More formally, a visibility pair
at a given position of~$\ell$ is a pair $(\rho,\beta)\in\R\times \B$ such that both $\rho$ and $\beta$ intersect $\ell$
and the segment that connects $\rho$ to $\beta$ along~$\ell$ does not intersect any other arc.)
Observe that for any purple intersection~$x\in \rho\cap\beta$, the pair
$(\rho,\beta)$ is a visibility pair just before the sweep line reaches~$x$,
so it is enough 
if we test a pair of curves for intersection
whenever it become a visibility pair. 

Maintaining the visibility pairs is done as follows. 
Let $\vd(\R)$ denote the vertical decomposition of the red arrangement~$\A(\R)$,
and let $\vd(\B)$ denote the vertical decomposition of the blue arrangement~$\A(\B)$.
With a slight abuse of terminology, we refer to the faces in these vertical decompositions 
as (red or blue) \emph{trapezoids}. A red trapezoid is called \emph{hot} if it is 
intersected by a blue arc, and a blue trapezoid is called hot if it is intersected by a red arc.
The algorithm sweeps over the two arrangements $\A(\R)$ and $\A(\B)$
in parallel---in other words, it performs two simultaneous sweeps that use the same sweep line~$\ell$---and 
while doing so, it constructs the hot trapezoids that intersect the sweep line.
(It does not construct the entire arrangements~$\A(\R)$ and~$\A(\B)$,
which would be too costly.)
Note that if $(\rho,\beta)$ is a visibility pair, then the vertical visibilities between 
$\rho$ and $\beta$ occur inside trapezoids that are hot in $\vd(\R)$ as well as in~$\vd(\B)$.
Thus, as explained by Har-Peled and Sharir~\cite{DBLP:journals/dcg/Har-PeledS01},
we can detect all visibility pairs---and, hence, all purple intersections---if we can
construct all hot trapezoids.

Har-Peled and Sharir then argue (Lemma~7.2) that the total number of hot blue trapezoids 
is $O(\lambda_{s+2}(m+k))$ and that they can be constructed in $O(\lambda_{s+2}(m+k)\log m)$
expected time in an online fashion during the sweep. In their argument, it is essential 
that the red arcs form a connected set. The red arcs are also connected in our setting, so we can 
implement the construction of the blue hot trapezoids in exactly the same manner. 
Har-Peled and Sharir can apply the same argument to the hot red trapezoids, because in
their setting the blue arcs also form a connected set. In our setting, the blue set need
not be connected. But since the red arcs
are the edges of a subdivision, the entire vertical decomposition $\vd(\R)$ has
$O(m)$ complexity.
Thus, the number of hot red trapezoids is not larger than for Har-Peled and Sharir. 
This suffices to reach the same time bound, as they present an argument 
that charges the time spent for the simultaneous sweep to hot red or blue trapezoids and purple intersections.

We conclude that the expected running time in our setting is the same
as in the setting of Har-Peled and Sharir, namely~$O(\lambda_{s+2}(m+k)\log m)$, which
proves the theorem.
\end{proof}
Using \autoref{thm:red-blue} to compute the conflict lists, we obtain the following result.
\begin{theorem}\label{thm:conflict-list-computation}
Let $\D$ be a set of $n$ objects from a family of objects with union complexity~$U(n)$.
Then we can compute in $O(U(n)\log^2 n)$ expected time a simplified stacking $\stackpir(\D)$, 
including an owner for each face of $\stackpir(\D)$, such that $|\stackpir(\D)|=O(U(n))$,
and $\sum_{f\in F} |K^*_f|=O(U(n)\log n)$ and  $\sum_{f\in F} |K_f|=O(U(n)\log n)$, 
where $F$ is the set of faces of $\stackpir(\D)$.
Moreover, we can compute all conflict lists $K_f$ in total expected time~$O( \lambda_{s+2}(U(n)) \log^2 n)$.
\end{theorem}
\begin{proof}
The most time-consuming part of the algorithm is computing the conflict lists. This is
done by applying the algorithm from \autoref{thm:red-blue} to $\stackpir(\D)$ and the
boundary arcs of the objects. Since the (expected) number of intersections is
$O(U(n)\log n)$ by \autoref{lem:size-of-conflict-sets}, the total running time
is $O( \lambda_{s+2}(U(n)\log n) \log n)$, which is $O( \lambda_{s+2}(U(n)) \log^2 n )$.
\end{proof}

\subsection{The SSSP algorithm}
Like all \sssp algorithms in unweighted graphs,
our algorithm performs a BFS-like search from the source object. Let $L_j$ be the set of objects
at distance~$j$ from the source, and define $L_{\leq j} := L_0 \cup\cdots\cup L_j$ and $L_{< j} := L_0 \cup\cdots\cup L_{j-1}$.
Given levels $L_0,\ldots, L_{j-1}$, our task is to compute~$L_{j}$, 
which contains the objects in $\D \setminus L_{< j}$ that intersect~$\union(L_{j-1})$.
Following earlier work~\cite{DBLP:conf/esa/BergC25,Klost23}, we will identify a set $\C_{j} \supseteq L_{j}$ 
of \emph{candidates} to be checked for intersection with~$L_j$. To control the
running time, we need to make sure that an object is in the candidate 
set $\C_j$ for only a few values of~$j$. How this is done is very different from the 
previous algorithms---this is where our stackings come into play. How we test the 
candidates for intersection with $L_j$ is facilitated by the stacking as well.
The following pseudocode gives an overview of our algorithm;
the details of Steps~\ref{step:loop-start}, \ref{step:Fj}, and~\ref{step:Lj} will be described later.
\begin{algorithm}[H] 
\caption{{\sc SSSP-for-Objects-with-Small-Union-Complexity}($\D,\dsource$)} \label{alg:sssp}
\begin{algorithmic}[1]
\State Compute a simplified stacking $\stackpir(\D)$ with its
       conflict lists, according to \autoref{thm:conflict-list-computation}. \label{step:init}
\State $L_0 := \{ \dsource \}$; \ \  $L_1 := \{ D\in \D \setminus \{ \dsource\} : D\cap\dsource \neq \emptyset \}$ 
\State $j := 2$; \ \ $\mbox{\emph{done}} := \mbox{{\sc false}}$ 
\While {not \emph{done}} 
    \State Compute $\union(L_{j-1})$, preprocess it for point location, and initialize $L_j := \emptyset$. \label{step:loop-start}
    \State $F_j := \{ f : \mbox{$f$ is a face of $\stackpir(\D)$ that intersects $\union(L_{j-1})$} \}$. \label{step:Fj}
    \For{each face $f\in F_j$}
    \State $\C_j(f) := K_f\setminus L_{< j}$ \hfill $\rhd$ the candidates in the face~$f$~\label{step:cand-list-f}
       \State \parbox[t]{123mm}{Determine all objects $D\in \C_j(f)$ that intersect $\union(L_{j-1})\cap f$ 
              and add each such object to $L_{j}$ (if it is not yet present in $L_j$).} \label{step:Lj}
    \EndFor \label{step:loop-end}
    \If{$L_{j}=\emptyset$} \label{step:if}
        \State  $L_{\infty} := \D \setminus L_{< j}$; \ \ 
                $\mbox{\emph{done}} := \mbox{{\sc true}}$ \hfill $\rhd$ nodes in $L_{\infty}$ are unreachable from $\dsource$ 
    \Else 
    \State $j := j+1$
    \EndIf
\EndWhile
\end{algorithmic}
\end{algorithm}

\begin{lemma} \label{lem:correctness}
{\sc SSSP-for-Objects-with-Small-Union-Complexity} is correct.
\end{lemma}
\begin{proof}
We prove by induction on $j$ that the sets $L_j$ are computed correctly. Trivially,
$L_0$ and $L_1$ are computed correctly, and setting $L_{\infty} := \D \setminus L_{< j}$
is correct when $L_j=\emptyset$. We now show that $L_j$ is computed correctly in
Steps~\ref{step:loop-start}--\ref{step:loop-end}, given that $L_0,\ldots,L_{j-1}$
have been computed correctly. The algorithm only adds objects to $L_j$ that are not in $L_{<j}$
and intersect an object in~$L_{j-1}$, and these must indeed be in~$L_j$.
It remains to argue that any object~$D$ that should be in $L_j$ is actually added.
Thus, we need to show that the set $F_j$ computed in Step~\ref{step:Fj}
contains a face~$f$  such that $D\in K_f$ and
$D$ intersects $\union(L_{j-1})\cap f$. 
Observe that $D$ cannot fully contain $\union(L_{j-1})$ since then $D$ would already
be in~$L_{<j}$. (Here we use that $j\geq 2$ in the while-loop.)
Hence, $\bd D$ intersects $\union(L_{j-1})$.
Let $p\in \bd D \cap \union(L_{j-1})$ and let $f$ be the face of $\stackpir(\D)$ that contains~$p$.
Observe that $f\in F_j$. Moreover, $D\in K_f$ and $D$ intersects $\union(L_{j-1})\cap f$,
which finishes the proof.
\end{proof}
Before we dive into the details of the algorithm,
we first make an observation that will be crucial to bound the overall running time.
\begin{observation} \label{obs:crucial}
Any face $f\in\stackpir(\D)$ appears in at most three of the sets~$F_j$ that are created
in Step~\ref{step:Fj} of \autoref{alg:sssp}.
\end{observation}
\begin{proof}
Let $j^*$ be the smallest index such that $f\in F_{j^*}$. Then $\owner(f)\in L_{j^*-1}\cup L_{j^*}$
and any object intersecting~$f$ is in $L_{\leq (j^*+1)}$. Hence,
$\union(L_{j})$ cannot intersect~$f$ for $j>j^*+1$.
\end{proof}
We also need the following folklore result on computing the union of a set of objects.
\begin{lemma}\label{lem:union-comp}
The union of a set $\D$ of $n$ constant-complexity objects in the plane with union
complexity $U(n)$ can be computed in $O(U(n)\log n)$ expected time.
\end{lemma}
\begin{proof}
This can be done using a standard randomized incremental (RIC)
algorithm~\cite{DBLP:journals/dcg/BoissonnatDSTY92,DBLP:journals/dcg/ClarksonS89,DBLP:books/daglib/0077673} that computes
a vertical decomposition of the complement of the union. This problem fits
into the RIC framework---any ``trapezoid'' in the vertical decomposition of the complement 
is defined by $O(1)$ objects, and a trapezoid appears in the final vertical
decomposition iff its interior is not intersected by any of the objects---and
so the RIC algorithm can be implemented to run in $O(U(n)\log n)$
expected time. (One easily checks that the so-called update conditions are satisfied.)
\end{proof}

\subparagraph{Step~\ref{step:loop-start}: Computing~$\union(L_{j-1})$.}
Computing $\union(L_{j-1})$ in Step~\ref{step:loop-start}
can be done in $O(U(|L_{j-1}|)\log n)$ expected 
time\footnote{Here and in the following
we write $\log n$ instead of, for example, $\log |U(|L_j|)$; this is justified
since the arguments of the $\log$-function will always be polynomial in~$n$,
and it simplifies the expressions.} by \autoref{lem:union-comp}.
Preprocessing $\union(L_{j-1})$ for $O(\log n)$-time point-location 
queries can be done in $O(U(|L_{j-1}|) \log n)$ time~\cite{bcko-cgaa-08}.
Thus, the expected time for Step~\ref{step:loop-start} is $O(U(|L_{j-1}|)\log n)$.

\subparagraph{Step~\ref{step:Fj}: Finding the faces that intersect~$\union(L_{j-1})$.}
Let $\Finj$ be the set of faces
that are contained in $\intr(\union(L_{j-1}))$ and let $\Fxj$ be the set of faces
that intersect~$\bd\hspace{0.2mm}\union(L_{j-1})$. 
Thus, $F_j = \Finj \cup \Fxj$. Let $\bar{F}_j$ denote the set
of faces $f$ whose conflict list~$K_f$ contains at least one object from~$L_{j-1}$.
We can assume that for each object $D$ we have maintained (during the construction of the conflict lists)
a list containing all faces in whose conflict list $D$ participates.
Using these lists, we can construct~$\bar{F}_j$
in $O(\sum_{D\in L_{j-1}} k_D)$ time, where $k_D$ is the number of conflict lists containing~$D$.
The set $\bar{F}_j$ contains all faces from $\Fxj$ plus, possibly, 
some faces from~$\Finj$.

Let $\G^*$ be the dual graph of $\stackpir(\D)$; to simplify the presentation,
we will not distinguish between the nodes in $\G^*$ and the faces of~$\stackpir(\D)$.
We do a multi-source DFS on $\G^*$,
which will visit exactly those faces that intersect~$\union(L_{j-1})$;
thus, it is not a complete DFS that visits all faces.
The source nodes for the DFS are the faces in $\bar{F}_j$.
Suppose that during the DFS we are visiting a face~$f$. 
For each edge $e$ of $f$, we must then inspect the face $f_e$ on
the opposite side of~$e$ and decide if we should visit~$f_e$. We do \emph{not} visit~$f_e$
if either (i) $f_e$ has already been visited,
or (ii) $e$ is contained in an edge of $\union(L_{j-1})$.
Observe that $f_e$ may intersect $\union(L_{j-1})$ even when Condition~(ii) holds; see \autoref{fig:dfs}(i) for an example.
But then $f_e\in \Fxj$, so this is not a problem.
\begin{figure}
    \centering
    \includegraphics{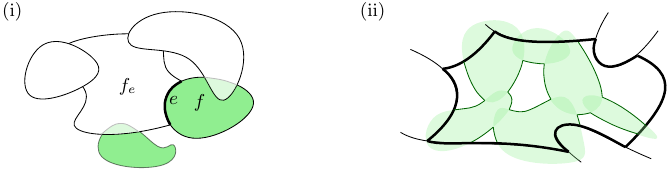}
    \caption{(i) The green objects form~$L_{j-1}$. Face $f_e$ is not visited through~$e$, but it is in $\Fxj$. \\
                (ii) The subdivision $\A_f$ of a face~$f$ (with $\bd f$ shown bold) induced 
                     by $\bd \union(L_{j-1}) \cap f$ (with $\union(L_{j-1})$ shown in green).
                     $\A_f$ has six regions, one of which is fully covered by $\union(L_{j-1})$.
                     The regions disjoint from $\union(L_{j-1})$ are all simply connected.}
    \label{fig:dfs}
\end{figure}

Condition~(i) is trivial to test.
Condition~(ii) can be tested in $O(1)$ time after performing a point-location query
with a point $p\in e$ in $\union(L_{j-1})$: if the point location
tells us that $p$ lies in the interior of $\union(L_{j-1})$ then we know that
$e$ is not contained in an edge of $\union(L_{j-1})$, and if the point location
tells us that $p\in e'$ for some edge~$e'$ of $\union(L_{j-1})$ then we only
need to check whether $e\subseteq e'$. Thus, the DFS takes time
$O(\sum_{f\in F_j} \deg(f)\cdot\log n)$, where $\deg(f)$ is the degree of~$f$ in~$\G^*$.
\begin{lemma} \label{lem:intersection-in-face}
Step~\ref{step:Fj} is correct and runs in time
$
O\left( \sum_{D\in L_{j-1}} k_D + \sum_{f\in F_j} \deg(f)\cdot\log n \right),
$
where $k_D$ is the number of conflict lists 
containing~$D$ and $\deg(f)$ is the degree of~$f$ in the dual graph of~$\stackpir(\D)$.
\end{lemma}
\begin{proof}
The running time of our algorithms follows from the discussion above.
It remains to argue that the algorithm is correct, that is, that all faces in $F_j$
are visited by the DFS. To see this, first observe that any face from $\Fxj$
is a source of the DFS. Next, observe that any face $f\in\Finj$ has a path in $\G^*$
to a face $f'\in\Fxj$ whose internal nodes correspond to faces in~$\Finj$.
(Simply follow a curve $\gamma\subset \union(L_{j-1})$ from any point $p\in f$ 
to a point on $\bd\hspace{0.2mm} \union(L_{j-1})$, and take $f'$ to be the first
face along $\gamma$ that is in $\Fxj$. This face exists since the last
face is in~$\Fxj$.)
Moreover, all edges incident to a face in $\Finj$ satisfy Condition~(ii).
Hence, we can conclude that the DFS indeed visits any face $f\in F_j$.
\end{proof}

\subparagraph{Step~\ref{step:Lj}: Finding the candidates that intersect~$\union(L_{j-1})\cap f$.}
We start by giving a short overview of Step~\ref{step:Lj}, and then explain it in detail.

We first construct the subdivision~$\A_f$ of $f$
induced by $\bd\hspace{0.2mm}\union(L_{j-1}) \cap f$; see \autoref{fig:dfs}(ii). 
Each region\footnote{We refer to the faces of $\A_f$ as \emph{regions} to avoid
confusion with the faces of $\stackpir(\D)$.}  
of $\A_f$ is contained in $\union(L_{j-1})$ or is disjoint from it. 
We also compute the set $X_f$ of intersection points between the
boundary arcs of the objects in~$\C_j(f)$ and~$\bd f$, using our red-blue
intersection algorithm. 
For each region $R\in \A_f$, we create a set $\D(R)\subseteq \C_j(f)$ of objects, as follows.
If $D\subset f$, then we take an arbitrary point $p_D\in \bd D\cap f$ and we put $D$ into the set $\D(R)$
of the region $R$ containing~$p_D$.
If $D\not\subset f$, then we put $D$ into the set $\D(R)$ of each region~$R$ 
such that $\bd D$ intersects $\bd R\cap \bd f$. Each such intersection is in $X_f$, 
so we already computed the necessary information.
Observe that if a region $R\in\A_f$ is contained in $\union(L_{j-1})$,
then obviously all objects in $\D(R)$ intersect~$\union(L_{j-1})$.
To handle the regions $R$ disjoint from $\union(L_{j-1})$, we run 
(for each region $R$ separately) our red-blue intersection algorithm 
on the red region~$R$ and the blue boundary arcs of the
objects in $\D(R)$. Whenever we detect an intersection between
a boundary arc of an object~$D$ and a boundary arc of $R$ that
borders a region~$R'$ that is contained in~$\union(L_{j-1})$,
we report that~$D$ intersects $\union(L_{j-1})$. We also remove $D$ from consideration,
to avoid that we report it too often.

The reason why this is correct is as follows. Suppose $D$ intersect~$\union(L_{j-1})\cap f$.
If $D\subset f$ and $p_D$ is in a region $R$ that is disjoint from $\union(L_{j-1})$, 
then clearly $\bd D$ must cross the boundary between $R$ and a neighboring region~$R'$
of $\A_f$, otherwise $D\subset R$ and $D$ cannot intersect~$\union(L_{j-1})\cap f$;
see \autoref{fig:intersection-conditions}(i).
If $D\not \subset f$ then $D$ enters one or more regions of $\A_f$
from outside~$f$;
see \autoref{fig:intersection-conditions}(ii).
\begin{figure}
    \centering
    \includegraphics{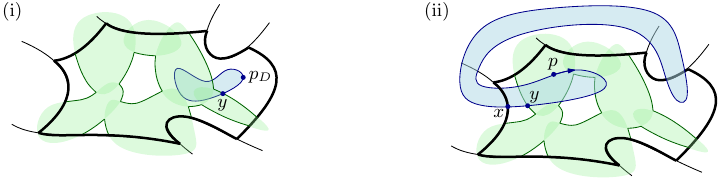}
    \caption{Illustration of how the subdivision $\A_f$ is used in Step~\ref{step:Lj} 
             to test whether an object $D \in K_f$ intersect $\union(L_(j-1)) \cap f$. 
             (i) $D$ is contained in~$f$.  (ii) $D$ intersects $\bd f$.}
    \label{fig:intersection-conditions}
\end{figure}
If all these regions are disjoint from $\union(L_{j-1})$ then, in order
to intersect~$\union(L_{j-1})$, $\bd D$ must cross the boundary between 
one such region~$R$ and a neighboring region~$R'$. Hence, all types of intersections are detected.

Recall that our red-blue algorithm required the red subdivision to be connected.
Thus, we must also argue that each region $R\in \A_f$ that is disjoint from~$\union(L_{j-1})$
is simply connected. To prove this, we use that the faces in a simplified stacking 
are simply connected, and that $\union(L_{<j})$ is connected.

To analyze the running time, we observe that (up to logarithmic overhead terms)
we can charge all the work to 
(i) the size of $\C_j(f)$, which is $O(|K_f|)$, and 
(ii) intersections between objects boundaries and face boundaries, of which there are $O(|K^*_f|)$, and
(iii) the complexity of $\A_f$, and
(iv) the complexity of~$f$. 
    
We now explain and analyze Step~\ref{step:Lj} in detail. 
\medskip

First, we construct~$\union(L_{j-1}) \cap f$, as follows. 
Let $L_{j-1}(f) := L_{j-1} \cap K_f$ be the set of objects in $L_{j-1}$ 
whose boundary intersects~$f$. We can obtain $L_{j-1}(f)$ in $O(|K_f|)$ time
by checking for each object in $K_f$ whether it is in $L_{j-1}$.
Next, we compute $\union(L_{j-1}(f))$ in 
$O(U(|L_{j-1}(f)|)\cdot\\log n)$ time by \autoref{lem:union-comp}.
We clip this union to within~$f$---that is,
we construct $\union(L_{j-1}(f)) \cap f$. Note that each intersection of
$\bd f$ and $\bd\hspace{0.2mm} \union(L_{j-1})$ contributes a point to~$K^*_f$, so 
we can do the clipping in $O( (U(|L_{j-1}(f)|) + |f| +|K^*_f|)\log n)$ time using a simple plane sweep algorithm.
Observe that $\union(L_{j-1}) \cap f = \union(L_{j-1}(f))\cap f$ unless there is an object
in $L_{j-1}$ that fully contains~$f$. We can test the latter in $O(\log n)$ time
by taking any point~$p\in f\setminus \union(L_{j-1}(f))$ and check if $p\in \union(L_{j-1})$
using point location.
(Recall that we already preprocessed $\union(L_{j-1})$ for point location queries in Steps~\ref{step:loop-start}.)
So far, the time we spent is
\begin{equation} \label{time1}
O\big( |K_f| + U(|L_{j-1}(f)|)\cdot\log n + (|K^*_f|+|f|)\log n \big) = 
O\big(  ( U(|K_f|) + |K^*_f|+|f|)\log n  \big).
\end{equation}
If $\union(L_{j-1})\cap f=f$ then all objects in $\C_{j}(f)$ intersect $\union(L_{j-1})$
and we are done. Moreover, when $\owner(f) \in L_{j-2}$ then $\C_j(f) = \emptyset$
because any object intersecting $f$ is then in $L_{<j}$. 
From now on, we therefore assume that
\begin{quotation}
 $\union(L_{j-1}) \cap f\neq f$ and $\owner(f) \not\in L_{j-2}$. \hfill $(\ast)$
\end{quotation}
We now need to determine the objects $D\in \C_j(f)$ 
that intersect $\union(L_{j-1})\cap f$. To this end, we first
construct the subdivision of $f$ induced by~$\bd\hspace{0.2mm}\union(L_{j-1})$. 
Given $\union(L_{j-1}) \cap f$, this is easily done in linear time. 
As mentioned earlier, we denote this subdivision by $\A_f$; see \autoref{fig:dfs}(ii).
Each region of $\A_f$ is either fully contained in $\union(L_{j-1})$ or it is
disjoint from it, and we label it accordingly. The subdivision~$\A_f$ will
come into play later.

After constructing~$\A_f$, we compute 
the set of intersections between the boundary arcs of the objects in~$\C_j(f)$ and the arcs comprising~$\bd f$.
We denote this set by~$X_f$. Note that $X_f \subseteq K^*_f$, where
$K^*_f$ is defined as in \autoref{lem:size-of-conflict-sets}.
($K^*_f$ and $X_f$ need not be the same, since
$X_f$ only considers objects $D\in \C_{j}(f)$, while $K^*_f$ considers all objects.) 
Since $f$ is simply connected, we can use \autoref{thm:red-blue} to do this
in $O(\lambda_{s+2}(|\C_j(f)|+|f|+ |K^*_f|) \log n)$ time. 
Since $\C_j(f) \subseteq K_f$, this is bounded by
\begin{equation} \label{time2}
O(\lambda_{s+2}(|K_f|+|f|+|K^*_f|)\log n).
\end{equation}
Now consider an object $D\in \C_j(f)$. We associate a set $\R_f(D)$ of regions 
to~$D$, as follows.\footnote{One can think of $\R_f(D)$ as the inverse of the ``mapping'' 
$\D(R)$ described intuitively above, whose formal definition will follow later.}
Let $X_f(D)$ be the set of intersection points of $\bd D$ with $\bd f$.
Thus, $X_f := \bigcup_D X_f(D)$.
\begin{itemize}
\item If $X_f(D) \neq \emptyset$ then
      we put all the regions of $\A_f$ into $\R_f(D)$ that contain at least one point
      of $X_f(D)$ on their boundary.
\item If $X_f(D)=\emptyset$ then we take an arbitrary point $p_D\in \bd D \cap f$
      and we put the region from $\A_f$ that contains~$p_D$ into~$\R_f(D)$.
\end{itemize}
The time to compute the sets $\R_f(D)$ is subsumed by the time needed to
compute the intersections between the boundary arcs of the objects in~$\C_j(f)$ 
and the arcs comprising~$\bd f$.
(Computing the region $R$ containing $p_D$ when $X_f(D)=\emptyset$ can be done using point location.)
The following lemma shows that 
the set $\R_f(D)$ limits the the number of regions we have to consider in
order to decide whether $D$ intersects $\union(L_{j-1}(f))\cap f$.
\begin{lemma} \label{lem:intersect-conditions}
An object $D\in \C_j(f)$ intersects $\union(L_{j-1}(f))\cap f$ iff either one of the following 
conditions holds:
\begin{enumerate}
\item[(i)] There is a region $R\in\R_f(D)$ that is fully contained in $\union(L_{j-1}(f))$.
\item[(ii)] There is a region 
      $R\in \R_f(D)$ such that $\bd D$ intersects $\bd R$ in a point $y\not\in X_f(D)$.
\end{enumerate}
\end{lemma}
\begin{proof} 
($\Longleftarrow$) Any region in $\R_f(D)$ contains a point from $D$ by definition, so 
$D$ intersects $\union(L_{j-1}(f))\cap f$ if Condition~(i) is met.
Furthermore, the point $y$ in Condition~(ii) must be on
a part of $\bd R$ that separates $R$ from a region fully contained 
in $\union(L_{j-1}(f))$, since $y\notin X_f(D)$ implies $y\not\in\bd f$.
Hence, $D\in C_j(f)$ intersects $\union(L_{j-1}(f))\cap f$ in this case as well.
\medskip

\noindent ($\Longrightarrow$) Suppose Condition~(i) is not satisfied.
We must show that then Condition~(ii) holds.
\begin{itemize}
\item If $X_f(D) \neq \emptyset$ then consider a point~$p\in D \cap \big( \union(L_{j-1}(f))\cap f \big)$.
      Starting at $p$, follow $\bd D$ (clockwise around $D$, say) until a point $x\in X_f(D)$ is reached,
      as in \autoref{fig:intersection-conditions}(ii).
      This point $x$ lies on the boundary of a region~$R$ that is disjoint from~$\union(L_{j-1}(f))$,
      otherwise Condition~(i) would be satisfied. Tracing $\bd D$ backwards
      towards~$p$, we must leave $R$ in a point $y\not\in X_f(D)$.
\item If $X_f(D) = \emptyset$ then the region $R$ containing the point~$p_D$
      referred to in the definition of $X_f(D)$ must be disjoint from $\union(L_{j-1}(f))$, since
      Condition~(i) is not satisfied. Since $\bd D$ does not intersect $\bd f$,
      this implies that $\bd D$ must intersect $\bd R$ in some point $y\not\in X_f(D)$,
      otherwise $D\subset R$ and $D$ cannot intersect $\union(L_{j-1}(f))$;
      see \autoref{fig:intersection-conditions}(i).
\end{itemize}
We conclude that Condition~(ii) holds in both cases, which finishes the proof.
\end{proof}
To find the objects $D\in \C_j(f)$ intersecting $\union(L_{j-1}(f))\cap f$, we must thus
determine which objects satisfy either of the Conditions~(i) or~(ii).

Condition (i) is trivial to test. To deal with Condition~(ii) we proceed as follows.
For each region~$R\in\A_f$ that is disjoint from~$L_{j-1}(f)$,
let $\D(R)$ be the set of objects $D\in \C_j(f)$ such that $R\in\R_f(D)$.
We must decide for each $D\in\D(R)$ whether $\bd D$ intersects $\bd R$
in a point $y\not\in X_f(D)$.
The crucial property ensuring we can do this efficiently, is the following.
\begin{lemma}
Any region $R\in\A_f$ that is disjoint from $\union(L_{j-1})$ is simply connected.  
\end{lemma}
\begin{proof}
Suppose for a contradiction that $R$ contains a hole~$h$. Since $f$ itself
is simply connected by the properties of a simplified stacking,
this implies that $h$ is at least partially contained in $\union(L_{j-1})$. (We say ``at least partially''
because $h$ might have a hole itself.) Since $j\geq 2$, there must then be an object
$D\in L_{j-2}$ that intersects~$h$. Since $\owner(f)$ is simply connected,
this implies that $\owner(f)$ intersects $D$, and so $\owner(f) \in L_{\leq j-1}$.
We cannot have $\owner(f)\in L_{<j-2}$, since then $\union(L_{j-1})$ cannot intersect~$f$.
Furthermore, $\owner(f)\in L_{j-2}$ directly contradicts Assumption~$(\ast)$.
Finally, $\owner(f)\in L_{j-1}$ contradicts Assumption~$(\ast)$ as well, since
$\owner(f)\in L_{j-1}$ implies that $\union(L_{j-1}) \cap f = f$.
\end{proof}
The above lemma implies that we can compute the intersections between
the boundary arcs of the objects in $\D(R)$ and $\bd R$ by applying the 
red-blue intersection algorithm from \autoref{thm:red-blue} 
to~$R$ and the set of arcs forming the boundaries of the objects in~$\D(R)$. 
This takes $O(\lambda_{s+2}(|R|+|\D(R)|+k_R) \log n)$ time,
where $|R|$ is the complexity of the region~$R$ and $k_R$ is the number of reported intersections.
In principle, an arc $\beta\subset \bd D$ can intersect $\bd R$ in too many points. 
(The intersections of $\beta$ with $\bd R \cap \bd f$ can be bounded 
using \autoref{lem:size-of-conflict-sets}, but for the intersections 
of $\beta$ with $\bd R\setminus \bd f$ this is not the case.)
Fortunately, this is not a problem: as soon as the plane-sweep algorithm 
from \autoref{thm:red-blue} detects an intersection~$y$
between some boundary arc $\beta\subset \bd D$ and $\bd R\setminus \bd f$, it can act as 
if $y$ is the right endpoint of~$\beta$ and delete~$\beta$. This is allowed since we now know
that Condition~(ii) holds for~$D$. 
Hence, $k_R$ is bounded
by the number of intersections of boundary arcs with $\bd R \cap \bd f$,
which is $\sum_{D\in \C_j(f)} |X_f(D)\cap R|$, plus~1. Recall that
$X_f = \bigcup_D X_f(D)$. Hence, the time needed to handle $R$ is
\[
O(\lambda_{s+2}(|R|+|\D(R)|+|X_f\cap R|) \log n).
\]
To obtain the total time over all regions $R$, we first observe that
$\sum_{R\in \A_f} |R| = O(|\A_f|)$. Furthermore, 
\[
\sum_{R\in \A_f} |\D(R)| \leq \sum_{D\in \C_j(f)} \left( |X_f(D)|+1 \right) = O(|K^*_f| + n_f ),
\]
where $n_f$ is the number of objects $D\in \C_j(f)$ that are fully contained in~$f$.
Finally, we note that $\bigcup_{R\in \A_f} |X_f\cap R| =O(|X_f|)=O(|K^*_f|)$.
Hence, the total time we spend over all regions~$R\in \A_f$ is
\begin{equation} \label{time3}
O(\lambda_{s+2}(|A_f|+|K^*_f|+n_f) \cdot \log n).
\end{equation}
By adding the bounds in (\ref{time1}), (\ref{time2}), and (\ref{time3}), we obtain the following result. 
\begin{lemma} \label{lem:find-intersections}
Step~\ref{step:Lj} runs in 
\[
O\big(U(|K_f|)\log n\big) + O\big( \lambda_{s+2}(|A_f|+|K_f|+|K^*_f|+n_f) \cdot \log n \big)  + O\big(|f| \log n\big))
\]
time per face $f\in F_j$.
\end{lemma}

\subparagraph{Putting everything together.}
To analyze the running time of the whole \sssp algorithm, we need to sum the 
the running times we obtained above for the various steps in the while-loop
over all the iterations of the loop. Using that \autoref{obs:crucial} 
implies that summing over all iterations does not blow up the running time,  
we obtain our main result, which we repeat here for convenience.
\maintheorem*
\begin{proof}
The correctness of our algorithm was proved in \autoref{lem:correctness}. We now
analyze the total running time.

The expected time for Step~\ref{step:init} is $O(\lambda_{s+2}(U(n)) \log^2 n )$ 
by \autoref{thm:conflict-list-computation}. The total time for Step~\ref{step:loop-start}
over all iterations is $\sum_j O(U(|L_j|)\cdot  \log n)$, which is bounded by $O(U(n) \log n)$
since the sets $L_j$ are disjoint and $U(\cdot)$ is at least linear. 
By \autoref{lem:intersection-in-face}, the total time for Step~\ref{step:Fj} is
\[
O\left( \sum_j \left\{ \sum_{D\in L_{j-1}} k_D + \sum_{f\in F_j} \deg(f)\cdot\log n \right\} \right),
\]
where $k_D$ is the number of conflict lists containing~$D$ and $\deg(f)$ is the degree of~$f$ in~$\G^*$.
Note that $\sum_j \sum_{D\in L_{j-1}} k_D = \sum_f |K_f| = O(U(n)\log n)$. Moreover,
$\sum_j\sum_{f\in F_j} \deg(f) = \sum_{f\in F} \deg(f)$ by \autoref{obs:crucial}, and
$\sum_{f\in F} \deg(f) = O(|F|)$ because $F$ is the set of faces of a planar graph.
We have $|F|=O(U(n))$ by \autoref{lem:size-of-stacking},
so the total time for Step~\ref{step:Fj} over all iterations is~$O(U(n)\log n)$.

It remains to analyze the total time for all Steps~\ref{step:Lj}. By \autoref{lem:find-intersections},
this is
\begin{equation} \label{eq:total}
\sum_j \sum_{f\in F_j}\left\{ O\big(U(|K_f|)\cdot \log n\big) + O\big( \lambda_{s+2}(|A_f|+|K_f|+|K^*_f|+n_f) \cdot \log n \big)  + O\big(|f| \log n\big)) \right\}  
\end{equation}
We have $\sum_j \sum_{f\in F_j} |K_f| = O(\sum_{f\in F} |K_f|)$ by \autoref{obs:crucial},
which is $O(U(n)\log n)$ by \autoref{lem:size-of-conflict-sets}. The same bound holds
for $\sum_j \sum_{f\in F_j} |K^*_f|$. Furthermore, every vertex of $\A_f$ for some $f\in F_j$
is either a vertex of $\union(L_{j-1})$, or a vertex of~$f$, or an intersection
point between $\bd f$ and $\bd D$ for some $D\in L_{j-1}$. Observe that such an intersection
is a point in $K^*_f$. Hence, 
\[
\sum_j \sum_{f\in F_j} |\A_f|= O\left(\sum_j \left\{ U(|L_j|) + \sum_{f\in F_j} (|f|+|K^*_f|)  \right\} \right).
\]
Finally, $\sum_j \sum_{f\in F_j} n_f = O(n)$ since each object can be contained in only
a single face of $f\in F$, and $\sum_j \sum_{f\in F_j} |f| = O(|\stackpir(\D)|) = O(U(n))$.
(Here we again use \autoref{obs:crucial}.) Combining everything, we see that (\ref{eq:total}) is bounded
by $O(\lambda_{s+2}(U(n))\log^2 n )$.

The bounds for pseudodisks and fat objects follow from the fact
that $U(n)=O(n)$ for pseudodisks~\cite{pseudodisk-union-86} and $U(n)=O(n \hspace{0.2mm} 2^{O(\log^* n)})$
for locally fat objects~\cite{AronovBES14}.
\end{proof}

\section{Almost exact diameter}
\label{sec:diamater}
We now turn our attention to the \diamp problem. As explained in the introduction, 
we will use the framework of Chang, Gao and Le~\cite{Chang0024}. This requires two
main ingredients: a near-linear \sssp algorithm and a so-called $r$-clustering.
The former was presented in the previous section, the latter will be presented below.

\subsection{Star-based $r$-clusterings}
Many graph problems can be solved efficiently using a suitable decomposition of the input graph.
For planar graphs, the $r$-divisions introduced by Frederickson~\cite{DBLP:journals/siamcomp/Frederickson87} 
are a widely-used example of such a decomposition.
An \emph{$r$-division} of a planar graph $\G=(V,E)$  is a pair $(\C,S)$
where $\C$ is a collection of (not necessarily disjoint) subsets of $V$, called \emph{regions},
and $S \subseteq V$ is a set of \emph{boundary nodes}, with the following properties.
\begin{enumerate}
\item Every node in $V$ belongs to at least one region $C \in \mathcal{C}$, that is, $\bigcup_{C\in \C} C = V$.
\item Every region $C \in \C$ contains at most $r$ nodes.
\item Let $C^{\circ} := C \setminus S$ be the set of \emph{interior nodes} of a region~$C\in \C$.
    Then a node $v\in C^{\circ}$ does not appear in any other region, and all of its neighbors are contained in $C$.
\end{enumerate}
\begin{lemma}[Lemma~2 from \cite{DBLP:journals/siamcomp/Frederickson87}] \label{lem: planar-r-division}
For any $n$-node planar graph, we can compute in $O(n \log n)$ time an $r$-division consisting of 
$O(n/r)$ regions with $O(\sqrt{r})$ boundary nodes each.
\end{lemma}

\noindent Intersection graphs do not admit $r$-divisions with the bounds from
\autoref{lem: planar-r-division} for $r= o(n)$, since they can have arbitrarily large cliques. 
This led Chang, Gao, and Le~\cite{Chang0024} to introduce \emph{clique-based $r$-clusterings},
where the boundary consists of a small number of cliques instead of individual nodes.
We generalize this to \emph{star-based $r$-clusterings}, as defined next.
\medskip

Let $\G=(V,E)$ be a connected graph.
Following Chang, Gao, and Le, we need the regions in our $r$-division of $\G$ to be connected. 
Thus, we define a \emph{cluster} to be a subset $C\subseteq V$ such that 
the induced subgraph~$\G[C]$ is connected, and we use the term \emph{$r$-clustering} 
instead of the term $r$-division. For a node $v\in V$, we define $\mystar(v)$ to be the closed neighborhood of~$v$.
With a slight abuse of terminology (since there can be edges between the neighbors
of~$v$) we refer to $\mystar(v)$ as a \emph{star}. 
Our clustering is \emph{star-based}, meaning that the cluster boundaries 
consist of stars. It is defined as 
follows; see \autoref{fig:star-based-clustering}.
\begin{figure}
\centering
\includegraphics{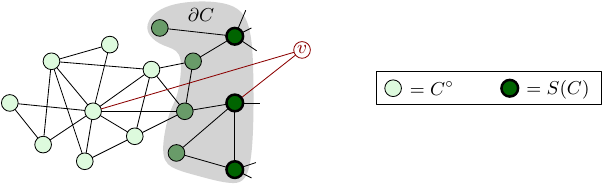}
\caption{Illustrating Property~\ref{prop: internal-to-outside} of a star-based $r$-clustering:
        any node $v\notin C$ with a neighbor in~$C^{\circ}$ must also have a neighbor in~$S(C)$.}
\label{fig:star-based-clustering}
\end{figure}

\begin{definition}\label{def: star-based-clustering}
    Let $\G=(V,E)$ be a connected graph with $n$ nodes and let $r$ be an integer with $1\leq r< n$. 
    A \emph{star-based $r$-clustering} of $\G$ consists of a set $\C$ of clusters 
    together with a set~$S(C)\subseteq C$ of star centers for each~$C \in \mathcal{C}$, such that the following holds:
    \begin{enumerate}
        \item \label{prop: everything assigned} Every node in $V$ belongs to at least 
               one cluster $C \in \mathcal{C}$, that is, $\bigcup_{C\in \C} C = V$.
        \item \label{prop: small clusters} Every cluster $C \in \C$ contains at most $r$ nodes.
        \item \label{prop: internal-to-outside} Define the boundary of a cluster $C$ as 
              $\bd C := C \cap \bigcup_{v\in S(C)} \mystar(v)$ and define its interior 
              as $C^\circ := C\setminus \bd C$. 
              If a node $v\not\in C$ has a neighbor in $C^\circ$ 
              then $v$ also has a neighbor in~$S(C)$.
    \end{enumerate}
    We  call $|\C|$ the \emph{size} of the clustering, we call
    $\max\{ |S(C)| : C\in \C \}$ the \emph{local boundary complexity} of the clustering, 
    and we call $\sum_{C \in \C} |S(C)|$ the \emph{global boundary complexity} of the clustering.
\end{definition}
In the remainder of this section we will prove the following theorem.
\begin{theorem}\label{thm: star-based-clustering}
Let $\D$ be a set of $n$ constant-complexity objects from a family of objects with  
union complexity~$U(n)$, such that any two boundary arcs intersect at most~$s$ times, 
and for which the intersection graph $\ig[D]$ is connected.
A star-based $r$-clustering of $\ig[D]$ of size $O(U(n) /\sqrt{r})$, local boundary complexity $O(\sqrt{r})$, 
and global boundary complexity $O(U(n)/\sqrt{r})$ can be computed in $O(U(n) \log ^2 n)$ expected time.
\end{theorem}
It is important to note that while the definition of a star-based
$r$-clustering refers to the stars of the nodes in~$S$, these stars
need not be computed explicitly: we only need to compute the 
set $\C$ of clusters and the set $S$ of star centers.

\subparagraph{\rm\em Remark.} Chang, Gao, and Le~\cite{Chang0024}
presented an algorithm to compute a clique-based $r$-clustering for a much more
restricted family of objects, namely constant-complexity pseudodisks that are similarly-sized and fat.
Besides the obvious fact that their cluster boundaries consist of cliques
instead of stars (which can be an advantage), they also have a stronger version of 
our Property~\ref{prop: internal-to-outside}: in their clustering,
an interior node of some cluster~$C$ does not have a neighbor outside~$C$.
Fortunately, our weaker property is sufficient for the \diamp algorithm
to be correct. On the positive side, we have an~$O(\sqrt{r})$ bound on
the local boundary complexity, while they only have a bound on the global boundary
complexity. This will help us to reduce the running time of the \diamp algorithm from
$\tilde{O}(n^{2-1/18})$ to $\tilde{O}(n^{2-1/14})$.

\subparagraph{Computing a star-based $r$-clustering.} 
We now present an algorithm that computes 
a star-based $r$-clustering for the intersection graph~$\ig[\D]$ of
a set $\D$ of $n$ constant-complexity objects with union complexity~$U(n)$. 
We assume that $\ig[D]$ is connected---this is easy to check by computing $\union(\D)$.
Our algorithm consists of the following steps.
\begin{enumerate}
    \item \label{step:cluster-init}
        Compute a stacking $\stackpi(\D)$ according to \autoref{lem:compute-stacking};
        we do not need the stacking to be simplified. Let $\G$ be the dual graph of $\stackpi(\D)$ (which is planar); 
        when convenient, we do not distinguish between a face of the stacking and its node in $\G$.
        \smallskip

        For each object~$D \in \D$, pick a representative point $\rep(D)$. If~$D$ owns at least one face, 
        then let~$\rep(D)$ be an arbitrary point inside a face owned by~$\rep(D)$. Otherwise, let~$\rep(D)$ be 
        an arbitrary point in~$D$ that is not on the boundary of a face.
    \item \label{step:augment} Augment $\G$ as follows. For every object $D$ that does not own any face,
        create a node $v_D$ and add an edge from $v_D$ to (the node corresponding to) the face 
        that contains~$\rep(D)$; see \autoref{fig:augment}. 
        \begin{figure}
        \centering
        \includegraphics{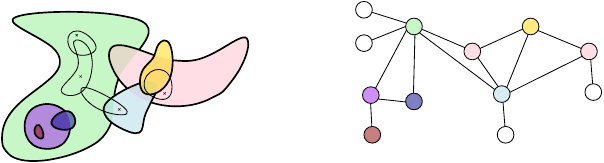}
        \caption{Left: a stacking of pseudodisks. There are four pseudodisks that do not own any face;
        their representative points are shown as crosses. Right: the augmented dual graph.}
        \label{fig:augment}
        \end{figure}
        Let $\G^*$ be the augmented graph, which is still planar.
        For each object~$D$, let~$V(D)$ be the set of nodes in~$\G^*$ that represent~$D$. So~$V(D)$ 
        contains the faces owned by~$D$ if there are any, and~$v_D$ otherwise.
    \item \label{alg: compute-r-division} 
        Compute an $r$-division $(\C^*,S^*)$ of $\G^*$ using \autoref{lem: planar-r-division}.
    \item \label{step:R}
       For each region $C^*\in\C^*$, create a subset $R(C^*)\subseteq \D$ that contains the objects
       corresponding to the nodes in~$C^*$. In other words,
       \[
       R(C^*) := \{ D \in \D : V(D) \cap C^* \neq \emptyset \}.
       \]
       We make an object $D\in R(C^*)$ a star center with respect to $R(C^*)$ 
       iff  $V(D) \cap S^*\cap C^* \neq \emptyset$. 
       Let $\R := \{ R(C^*) : C^* \in \C^* \}$ be the resulting collection of subsets.
       Observe that an object~$D$ can be
       a star center with respect to one set $R\in \R$ while $D$ is not a star center 
       with respect to some other set $R'\ni D$.
       This is necessarily to control the (local and global) boundary complexity.
    \item \label{step:region-to-cluster} For each $R \in \R$, compute the set~$\C(R)$ of connected components of $\ig[R]$.
          Let $\C$ be the set of connected components computed over all $R\in\R$.
          We report $\C$, with for each cluster $C\in\C$ its set $S(C)$ of star centers,
          as our star-based $r$-clustering.
\end{enumerate}
\medskip

\noindent Proving that the star-based $r$-clustering computed by the above algorithm has
Properties~\ref{prop: everything assigned} and~\ref{prop: small clusters} 
from \autoref{def: star-based-clustering} is straightforward.
The proof that it has Property~\ref{prop: internal-to-outside} is based on the fact
that every face in the stacking $\stackpi(\D)$ has an owner. Intuitively,
we can use this as follows. Suppose there is a cluster $C_1$ such that an object $D_2\not\in C_1$ intersects an
object~$D_1\in C_1^{\circ}$. 
Now trace a curve $\gamma\subset D_1\cup D_2$ through~$\stackpi(\D)$,
from the face responsible for adding $D_1$ to $C_1$, to the relevant face for $D_2$.
Consider the path in~$\G^*$ corresponding to the visited faces.
This path must contain a node in $S^*$ that is a face whose
owner is in $S(C_1)$ and that intersects~$D_2$ (which proves Property~\ref{prop: internal-to-outside})
or $D_1$ (which contradicts that $D_1\in C_1^{\circ}$).

Before we give a detailed proof that all three properties are indeed satisfied, 
we analyze the running time of the algorithm.
\begin{lemma} \label{lem:clustering-runtime}
The algorithm above runs in $O(U(n)\log^2 n)$ expected time.    
\end{lemma}
\begin{proof}
Step~\ref{step:cluster-init} runs in $O(U(n)\log^2 n)$ expected time by \autoref{lem:compute-stacking}.
Since \autoref{lem:compute-stacking} guarantees that the size of $\G^*$ is $O(U(n))$, Step~\ref{alg: compute-r-division}  runs in $O(U(n)\log n)$ time.
The times for Steps~\ref{step:augment} and~\ref{step:R} are subsumed by this.
(Finding the face containing $\rep(D)$ in Step~\ref{step:augment} can be done by point location.)
Finally, Step~\ref{step:region-to-cluster} can be done in $O(|R| \log^2 |R|)$ time 
for each $R\in \R$, since it boils down to computing $\union(R)$.
Thus, the total time for Step~\ref{step:region-to-cluster} is $\sum_{R\in\R} O(|R| \log^2 |R|)$,
which is $O(U(n)^2\log n)$ because $\sum_{R\in\R}|R| = O(\sum_{C^*\in\C^*}|C^*|) =O(U(n))$.
\end{proof}
%
\subparagraph{Proof of correctness.} 
We now prove that $\C$ fulfills the properties of a star-based $r$-clustering.
\begin{lemma}\label{lem: star-prop}
The above algorithm computes a valid star-based $r$-clustering.
\end{lemma}
\begin{proof}
    It follows directly from the construction that $\C$ satisfies Properties~\ref{prop: everything assigned} 
    and~\ref{prop: small clusters}. 
    To establish Property~\ref{prop: internal-to-outside}, consider an object $D_1$ that is an interior node in a cluster $C_1\in \C$
    and an object $D_2\not\in C_1$ that is a neighbor of $D_1$ in $\ig[\D]$. We will prove that~$D_2$ has a neighbor in~$S(C_1)$. 
    
    Let~$C^*_1 \in \C^*$ be the region in the $r$-division~$(\C^*, S^*)$ 
    that contributed cluster~$C_1$ as a connected component of~$\ig[R(C^*_1)]$ in Step~\ref{step:region-to-cluster}. 
    Since~$D_1 \in C_1$, there is a node~$f_1 \in V(D_1)$ representing it in~$\G^*$ with
    $f_1 \in C^*_1$. Let~$f_2 \in V(D_2)$ be an arbitrary node of~$\G^*$ that represents~$D_2$. 
    If~$f_2 \in C^*_1$, then (as~$D_2$ intersects~$D_1$) it would belong to the 
    same connected component of~$\ig[R(C^*_1)]$ and therefore to~$C_1$. Hence~$f_2 \notin C^*_1$. 
    
    We now claim the following: there is a (not necessarily simple) path~$\pi$ in~$\G^*$ from~$f_1$ to~$f_2$, 
    such that all internal nodes on the path correspond to faces of~$\stackpi(\D)$ (i.e., nodes of~$\G$) that 
    intersect~$D_1$ or~$D_2$ (or both). If~$f_1$ and~$f_2$ represent faces of~$\stackpi(\D)$, 
    we can argue the existence of~$\pi$ as follows. Take any curve~$\gamma$ that starts 
    in the interior of face~$f_1$ and ends in the interior of face~$f_2$ and that lies fully in~$D_1 \cup D_2$; 
    such a curve exists since $D_1$ and $D_2$ own $f_1$ and $f_2$ (and, hence, fully cover them) and $D_1$ intersects $D_2$. Since~$\G$ is the dual graph of~$\stackpi(\D)$, a path~$\pi$ with the claimed properties can then be obtained 
    as the sequence of faces intersected by~$\gamma$; see \autoref{fig:gamma-new}. 
    \begin{figure}
    \centering
    \includegraphics{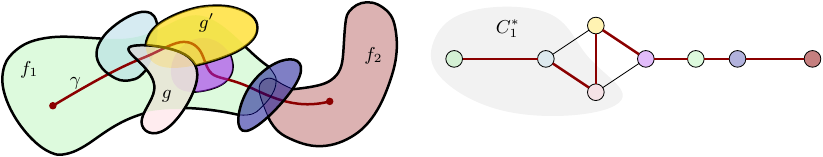}
    \caption{Illustration for the proof of \autoref{lem: star-prop}. Left: the curve~$\gamma\subset D_1\cup D_2$
             from $f_1$ to~$f_2$. Right: the corresponding path~$\pi$ in the dual graph~$\G^*$.
             If $g$ is the first face whose successor is not in~$C^*_1$, then the object corresponding to $g$
             is a star center in~$C_1$.}
    \label{fig:gamma-new}
    \end{figure}
    If node~$f_1$ does not represent a face of~$\stackpi(\D)$ but is the unique node~$v_{D_1}$ in~$\G^*$ representing~$D_1$, then we start curve~$\gamma$ at point~$\rep(D_1)$ instead; note that~$v_{D_1}$ is adjacent in~$\G^*$ to the node for the face containing~$\rep(D_1)$ and that the owner of that face intersects~$D_1$. Similarly, if~$f_2$ does not represent a face of~$\stackpi(\D)$ then we end~$\gamma$ at~$\rep(D_2)$ instead. We obtain the internal vertices of path~$\pi$ as the faces visited by~$\gamma$. 

    Using path~$\pi$ we complete the proof. Let~$g$ be the first node of~$\pi$ whose successor~$g'$ does not belong to~$C^*_1$. 
    Since~$f_1$ belongs to~$C^*_1$ but~$f_2$ does not, this is well-defined. As~$g$ itself belongs to~$C^*_1$, we have~$g \neq f_2$. 
    The fact that~$g \in C^*_1$ has a neighbor~$g'$ in another region implies that~$g \in S^*$: it is a boundary node of the $r$-division. 
    If~$f_1 = g$, then the vertex~$f_1$ representing~$D_1$ in~$C^*_1$ belongs to~$S^*$, which would imply that~$D_1 \in S(C_1)$. 
    But that means~$D_1$ is not an interior node of~$C_1$; a contradiction. Hence~$f_1 \neq g$, which means~$g$ is an internal node on path~$\pi$ that 
    represents a face of~$\stackpi(\D)$. Our choice of~$\pi$ guarantees that~$g$ intersects~$D_1$ or~$D_2$. To complete the proof, we distinguish these two cases.
    \begin{itemize}
        \item Suppose~$g$ intersects~$D_1$. Since~$g \in C^*_1$, the object~$\owner(g)$ is added to~$R(C^*_1)$ during Step~\ref{step:R}. Object~$D_1$ represented by~$f_1$ is added as well. 
        Then the intersecting objects~$\owner(g)$ and~$D_1$ end up in the \emph{same} connected component of~$\ig[R(C^*_1)]$ 
        in Step~\ref{step:region-to-cluster}, which is the cluster~$C_1$. Since~$g$ is a boundary node of~$C^*_1$, 
        this means~$\owner(g)$ is added to~$S(C_1)$ in the algorithm (that is, it is a star center with respect to $C_1$).
        Since~$\owner(g)$ belongs to~$S(C_1)$ and is adjacent to~$D_1$ in~$\ig[\D]$, this contradicts that~$D_1$ is an internal node of~$C_1$.
        \item Suppose~$g$ intersects~$D_2$. We claim that for each node~$x$ on the prefix~$\pi'$ of~$\pi$ that ends with~$g$, the object~$\owner(x)$ is contained in~$C_1$. 
        To see this, note that each such node belongs to~$C^*_1$ by choice of~$g$ and that 
        consecutive nodes in the dual graph of~$\stackpi(\D)$ 
        represent two intersecting objects. Hence, all such objects belong to the \emph{same} connected component of~$\ig[R(C^*_1)]$, which means all these objects belong to the resulting cluster~$C_1$. Consequently,~$\owner(g) \in C_1$. As~$g$ is a boundary node of the $r$-division, this implies~$\owner(g) \in S(C_1)$. Since~$g$ intersects~$D_2$, this establishes that Property~\ref{prop: internal-to-outside} holds.\qedhere
    \end{itemize}
\end{proof}

To finish the proof of \autoref{thm: star-based-clustering}, it remains to bound the size and the
local and global boundary complexity of the computed star-based $r$-clustering.
\begin{lemma}\label{lem: star-size}
    The star-based $r$-clustering $\C$ 
    has size $O(U(n)/\sqrt{r})$, 
    local boundary complexity $O(\sqrt{r})$, and global boundary complexity 
    $O(U(n) /\sqrt{r})$.
\end{lemma}
\begin{proof} 
    Each region $C^*$ in the $r$-division $(\C^*,S^*)$ computed in 
    Step~\ref{alg: compute-r-division} contains $O(\sqrt{r})$ boundary nodes by \autoref{lem: planar-r-division}.
    Hence, after Step~\ref{step:R} each subset in~$R(C^*)$ contains $O(\sqrt{r})$ 
    star centers. 
    A cluster generated in Step~\ref{step:region-to-cluster} for some subset $R(C^*)$
    obviously cannot contain more boundary nodes than~$R(C^*)$.
    Thus, the local boundary complexity of $\C$ 
    is $O(\sqrt{r})$.
    \medskip

   Observe that the augmented graph~$\G^*$ has $O(U(n))$ nodes.
    By \autoref{lem: planar-r-division}, $(\C^*,S^*)$ thus has 
    $O(U(n)/r)$ regions, each with $O(\sqrt{r})$ nodes in $S^*$, so
    $\sum_{C^* \in \C^*} |S^*\cap C^*| = O(U(n)/\sqrt{r}).$ 
    By construction, the number of star centers in $R(C^*)$ is equal to $|S^*\cap C^*|$.
    Because the clusters generated 
    in Step~\ref{step:region-to-cluster} for a given subset $R(C^*)\in \R$ are disjoint,
    it must hold that $\sum_{C\in \C(R)} |S(C)|=|S^*\cap C^*|$,
    where $\C(R)$ is the set of clusters generated for~$R$.
    Because $\C = \bigcup_{R \in \R} \C(R)$, we thus have
    \[
    \sum_{C \in \C} |S(C)| = \sum _{R \in \R} \sum_{C\in \C(R)} |S(C)| = 
    \sum_{C^*\in\C^*} |S^*\cap C^*| = O(U(n) / \sqrt{r}),
    \]
    proving the bound on the global boundary complexity.
    \medskip
    
    It remains to show that the size of $\C$ is $O(U(n) /\sqrt{r})$.
    We do this by proving the following claim: each cluster contains a star center.
    Since the global boundary complexity is  $O(U(n) /\sqrt{r})$, this claim implies
    that the size of $\C$ is $O(U(n) /\sqrt{r})$ as well.
    
    To prove the claim, consider a cluster $C$ generated for a region $R \in \R$. 
    Recall that $C$ is a connected component of~$\ig[R]$. Because $r<n$ 
    we have $|R|<n$ and, hence, $|C| <n$. Since $\ig[\D]$ is connected,
    there must be some node $v \not \in C$ that has an edge to some node $u \in C$. 
    If $u\in \bd C$ then, by definition, $u\in \mystar(s)$ for some $s\in S(C)$.
    Thus, $C$ contains a star center, as claimed.
    If $u\not\in \bd C$ then $v$ has a neighbor in~$S(C)$, by Property~\ref{prop: internal-to-outside}
    of a star-based $r$-clustering, which also proves the claim.
\end{proof}

\subparagraph{\rm\em Remark.} Define a \emph{star-based $r$-division} to be a star-based
$r$-clustering $\C$ without the condition that the subsets in $\C$ induce
a connected subgraph. In other words, we can work with regions instead of clusters.
Then Step~\ref{step:region-to-cluster} is not needed, which implies that 
we only have $O(n/r)$ regions, like in Frederickson's $r$-division for planar graphs.
This improved bound on the size of~$\C$ may be useful in applications where
connectedness of the regions is not required.

\subsection{The \diamp algorithm}
We now present an algorithm to approximate 
the diameter of the intersection graph $\G := \ig[\D]$ of a set $\D$ of pseudodisks, up to an additive error of~2.
Note that the objects in $\D$ must be pseudodisks;
it is not sufficient if their union complexity~$U(n)$ is near-linear.
Also recall that $U(n)=O(n)$ for pseudodisks.
In the following, it will be convenient to denote the node set of $\ig[\D]$ by $V$
(even though $V=\D$).
\medskip

Let $d(u,v)$ denote the hop distance between two nodes $u,v$ in $\G$.
Recall that the \emph{eccentricity} of a node $u$ is defined as 
$\ec(u) := \max_{v\in V}  d(u,v)$. Thus, $\diam(\G) = \max_{u\in V} \ec(u)$
and it is sufficient to approximate the eccentricity of each node~$u\in V$ 
up to an additive error of~2. 
Following the framework of Chang, Gao, and Le~\cite{Chang0024}
we will approximate the eccentricities using an (in our case star-based) $r$-clustering
and using so-called patterns, as introduced by Fredslund-Hansen, Mozes, and Wulff-Nilsen~\cite{DBLP:conf/isaac/Fredslund-Hansen21}. 
The concept of patterns has also been used in earlier works~\cite{DBLP:conf/stoc/LiP19,WULFFNILSEN2013831}.

\subparagraph{Patterns.}
Let $Q = \langle q_0,\dots,q_{k-1} \rangle$ be a sequence of (not necessarily neighboring)
nodes in a given graph~$\G$. The \emph{pattern of a node $v$ with respect to $Q$} is the 
$k$-dimensional vector $\pt[v,Q]$ such that
$\pt[v,Q][i] = d(v,q_i)-d(v,q_0)$ for every $0 \leq i \leq k-1$.
The distance between some node $u$ and a pattern $\pt[v,Q]$ is defined as $d(\pt[v,Q],u):=\min_i \{ d(u,q_i)+\pt[v,Q][i]\}$.

Observe that if two nodes $u,v$ have the same pattern with respect to~$Q$, 
then $d(u,q_i) = d(v,q_i) + \delta$ for all~$i$, where $\delta:=d(u,q_0)-d(v,q_0)$. 
Hence, either all distances from $u$ to the nodes in~$Q$ are the same as the distances from~$v$,
or they are all larger, or they are all smaller.
This can be used to speed up the diameter computation for a planar graph~$\G$, as follows. 
Let $C$ be a region in an $r$-division $(\C,S)$ of~$\G$, and suppose we want to determine
the node $v\not\in C$ that is furthest away from any of the nodes in~$C$.
Then, by the observation just made, we only need to consider one node for each pattern
with respect to $Q:=C \cap S$, namely the node $v$ with the largest distance to that pattern.
This is true because any path from a node outside $C$ to a node in~$C$ must pass through a node in $Q$.

In our setting, we have a pseudodisk graph and we will use a star-based $r$-clustering~$\C$.
For each cluster~$C$ we take $Q$ to be an arbitrary ordering of the nodes in~$S(C)$.
The properties of a star-based $r$-clustering do not guarantee that any path
from a node outside~$C$ to a node in $C$ must pass through a node in~$S(C)$ but
observe that for nodes $u$ and $v$ for which a shortest path $\pi(u,v)$ 
goes through a node (which is a star center) in $Q$, we have $d(u,v)=d(\pt[v,Q],u)+d(v,q_0)$. 
When $\pi(u,v)$ visits a node adjacent to a node in~$Q$,
then we have a slightly weaker guarantee, as stated in the following lemma.
The lemma is similar to Lemma~16 by Chang, Gao, and Le~\cite{Chang0024}.
\begin{lemma}\label{lem: approx-path-boundary}
    Let $Q$ be an arbitrary sequence of nodes in a graph~$\G$.
    \begin{enumerate}[(i)]
    \item For any two nodes $u,v\in V$ we have $d(u,v) \leq d(\pt[v,Q], u) + d(v,q_0)$. \label{approx-i}
    \item If there is a shortest path $\pi(u,v)$  that contains a neighbor of 
          a node in $Q$, then we additionally have $d(\pt[v,Q],u) + d(v,q_0) \leq d(u,v)+2$.  \label{approx-ii}
    \item If there is a shortest path $\pi(u,v)$  that contains a node in $Q$, then 
         $d(\pt[v,Q],u) + d(v,q_0) = d(u,v)$.  \label{approx-iii}
    \end{enumerate}
\end{lemma}
\begin{proof}
    Because $d(\pt[v,Q],u) + d(v,q_0)$ is the length of a valid path between $u$ and $v$---namely
    a path via a node in~$Q$---we have $d(u,v) \leq d(\pt[v,Q],u) + d(v,q_0)$.
    This proves~(i).
    
    Now suppose $\pi(u,v)$ contains a node~$w$ that has a neighbor $q \in Q$. 
    Note that $d(\pt[v,Q],u) + d(v,q_0) \leq d(u,q) +d(q,v)$ 
    because $q \in Q$. Because $w$ and $q$ are neighbors, we have
    \[
    d(u,q) +d(q,v) \leq d(u,w) + d(w,v) +2 = d(u,v)+2,
    \]
    which proves 
    part~(ii) of the lemma. Part~(iii) follows in a similar manner.
\end{proof}
To be able to make use of the above lemma we need the following result.
\begin{lemma}\label{lem:inside-to-outside-cluster}
Let $\pi$ be a path in $\G$ that starts at a node~$u$ in a cluster~$C$
and that visits a node~$v\not\in C$. Then there is a node $s \in S(C)$ 
that is a node in $\pi$ or that has a neighbor in~$\pi$.
\end{lemma}
\begin{proof}
    Consider the first node $w$ of $\pi$ that does not belong to $C$; 
    such a node exists since $v \notin C$. Let $w'$ be the predecessor of $w$, 
    which therefore belongs to $C$. If $w' \in \bd C$ then, by definition,
    $w'$ is in $S(C)$ itself or $w'\in \mystar(s)$ for some $s\in S(C)$.
    Hence, we are done in this case. If $w' \notin \bd C$ then, by Property~\ref{prop: internal-to-outside}, 
    node~$w$  has a neighbor in $S(C)$, and we are also done.
\end{proof}
\autoref{lem:inside-to-outside-cluster} can be used to argue that patterns can
provide us with a good approximation of distances. To use patterns efficiently, however,
we need that there are not too many patterns for a given sequence~$Q$.
Chang, Gao, and Le~\cite{Chang0024} prove that this is indeed the case
for pseudodisk graphs. They show (Theorem~14) that the so-called distance 
encoding VC-dimension of pseudodisk graphs is~4---this is where we need
the objects to be pseudodisks and not just any family with near-linear union complexity.
They then combine this
with earlier results on the number of different patterns when the VC-dimension
is bounded (Lemma~15). This results in the following lemma.
\begin{lemma}[combines Theorem 14 and Lemma 15 from~\cite{Chang0024}] \label{lem: pattern}
    Let $\D$ be a set of pseudodisks and $Q$ be a sequence of $k$ vertices in $\ig[\D]$. 
    Let $\cP := \{\pt[v,Q] : v \in V\}$ be the set of all patterns with
    respect to $Q$. If $d(q_0,q_i) \leq \Delta$ for every $1\leq i \leq k-1$, 
    then $|\cP|=O((k\cdot \Delta)^4).$
\end{lemma}

\subparagraph{The algorithm.}
Our algorithm is the same as that of Chang, Gao, and Le~\cite{Chang0024}, except that
we replace their clique-based $r$-clustering with our star-based $r$-clustering
and replace their \sssp algorithm with ours. 

The additive error of~2 in the computation
does not change: they can pick any node in a clique as a representative, while we must take
the center of a star, but the important property is that this representative has distance~1
to all other nodes in the clique or star. This is enough to give an additive error of~2;
see our \autoref{lem: approx-path-boundary} and their Lemma~16.
We note that the improvement to an additive error of~1, which is described
in the arXiv version of their paper~\cite{DBLP:journals/corr/abs-2401-12881},
does not apply when working with stars instead of cliques.

Our running time is better than theirs, as explained next.
Recall that we use the nodes in~$S(C)$ to define the patterns for a cluster~$C$.
Our clusters have local boundary complexity~$O(\sqrt{r})$ and size~$r$.
By \autoref{lem: pattern}, this means that the number of different patterns
for a cluster is $O((r\sqrt{r})^4)= O(r^6)$. They only have a bound on the
global boundary complexity, so they obtain a bound of $O((r\cdot r)^4)=O(r^8)$
on the number of patterns. A careful analysis shows that this improves the
running time from~$\tilde{O}(n^{2-1/18})$ to~$\tilde{O}(n^{2-1/14})$.

To keep the paper self-contained, we next describe and analyze the algorithm in detail.
\medskip

Recall that the goal of the algorithm is to compute the eccentricity of each node. 
We start by giving the intuition behind the algorithm. 

We use our star-based $r$-clustering $\C$ on the intersection graph~$\G := \ig[\D]$, 
which allows us to distinguish two types of shortest paths:
shortest paths that are fully contained in one cluster $C\in\C$, 
and shortest paths that leave the cluster. 
To handle the first type we can simply do the following for each cluster~$C$:
run the \sssp algorithm from \autoref{thm:sssp} on $\ig[C]$ with each $v\in C$ as the source.
Because the size of each cluster is $O(r)$ and we have $O(n/\sqrt{r})$ clusters, 
this takes $\tilde{O}(r^2 \cdot (n/\sqrt{r})) = \tilde{O}(n^2/\sqrt{r})$ time.

For the other type we use patterns. We first compute, for each node $v$ and 
each $Q(C)$ the pattern $\pt$. Here, $Q(C)$ denotes an arbitrary ordering of $S(C)$. This requires us to run our \sssp algorithm 
for each node in every sequence $Q(C)$. Since $\sum_{C\in\C}|S(C)|=O(n/\sqrt{r})$
we can do this in $\tilde{O}(n^2/\sqrt{r})$ time in total.  

After computing the patterns
$\pt$ for every $v\in V$ and $C\in \C$, we can approximate distances achieved by
paths of the second type. Using that the number of different patterns of a cluster
is $O(r^6)$ for each cluster~$C$ and the bounds from \autoref{thm: star-based-clustering},
we can show that the total time to handle paths of the second type is $O(r^7 (n/\sqrt{r}))$.
Setting $r=n^{1/7}$ then gives the final result.
\medskip

\noindent The following algorithm formalizes the intuitive description given above.
\begin{enumerate}
    \item \label{step:diam1} \begin{enumerate}
          \item Compute a star-based $r$-clustering $\C$ of~$\G := \ig[\D]$, 
                according to \autoref{thm: star-based-clustering}. For each star center $s \in \bigcup_{C \in \C} S(C)$, 
                run the \sssp algorithm from \autoref{alg:sssp} on $\G$ with source $s$. 
           \item \label{step:diam1b} For each cluster $C \in \C$, do the following. Create a sequence $Q(C)$ by arbitrarily ordering the nodes of $S(C)$. 
                 For each node $v \in V$, generate the pattern $\pt$ and store these patterns in a set~$\cP_C$.
                 (Recall that we already computed the distance from $v$ to each $s\in Q(C)$,
                 which is all that is needed to generate the patterns.) 
                 Sort the set of patterns, remove duplicates, and keep for each node~$v$ a pointer
                  to its pattern in~$\cP_C$. Finally, compute all pairwise distances in $\ig[C]$
                  by  running the \sssp algorithm $|C|$~times on $\ig[C]$, once for each $u\in C$ as the source.
            \end{enumerate}
    \item \label{step:diam2} 
          For each cluster $C \in \C$, each pattern $P \in \cP_C$, and each node $u \in C$, 
          compute the distance $d(P,u) :=\min_i \{ d(u,q_i)+P[i]\}$ of $u$ to the pattern~$P$.
          Let $u(P,C)$ be a node $u\in C$ that maximizes $d(P,u)$.
    \item \label{step:diam3} For each node $v$ we compute an approximation $\widetilde{\ec}(v)$ of $\ec(v)$.
          To this end, we compute for each cluster~$C\in \C$ a value $\Delta(v,C)$ that
          approximates the largest distance from~$v$ to any node in $C$. More precisely, 
          $\Delta(v,C)$ will satisfy
          \begin{equation} \label{eq:Delta}
          \max_{u\in C} d(v,u) \leq \Delta(v,C) \leq \max_{u\in C} d(v,u) + 2.
          \end{equation}
          After computing all values $\Delta(v,C)$, we set $\widetilde{\ec}(v):=\max_{C\in \C}\{\Delta(v,C)\}$.
          \smallskip
          
          Computing $\Delta(v,C)$ is done as follows. 
          Let $Q(C)$ be the sequence created for $C$ in Step~\ref{step:diam1} 
          and let $q_0$ be the first element of $Q(C)$. 
          Let $P = \pt$ be the pattern of $v$ with respect to~$Q(C)$.
          We distinguish two cases. 
          \smallskip
          
          \emph{Case (a): $v \not \in C$.} 
           In this case the node~$u(P,C)$
          computed in Step~\ref{step:diam2} is the node in~$C$ that is (almost) furthest from $v$,
          so we set $\Delta(v,C) := d(v,q_0)+d(P,u(P,C))$. 
          \smallskip
          
          \emph{Case (b): $v \in C$.} In this case a shortest path between $v$ and a vertex $u \in C$
           is either completely contained in $C$ or
          it intersects the neighborhood of~$S(C)$ (by \autoref{lem:inside-to-outside-cluster}).
          For each $u\in C$ we take the shortest of these two options, 
          and then we take the maximum over~$u$. Thus, we set
          \[
          \Delta(v,C)=\max_{u \in C} \big\{\min\{d_{\ig[C]}(u,v), d( v,q_0)+d(P,u) \} \big\},
          \]
          where $d_{\ig[C]}(u,v)$ is the distance between $u$ and $v$ in $\ig[C]$.   
\end{enumerate}
\medskip

\noindent We now prove \autoref{thm:diam}, which we restate for convenience,  by analyzing
the approximation error and running time of the above algorithm.
\diamtheorem*
\begin{proof}
    We first analyze the approximation error and then the running time.
    \medskip

    Our algorithm computes an approximate eccentricity $\widetilde{\ec}(v)$ for each $v\in V$.
    To bound the approximation error, it thus suffices to show that $\ec(v) \leq \widetilde{\ec}(v) \leq \ec(v)+2$.   
    To show this, it suffices to show that Inequality~(\ref{eq:Delta}) holds for all pairs~$(v,C)$.
    Let $P := \pt$ be the pattern of $v$ with respect to $Q(C)$.
    As in Step~\ref{step:diam3} of the algorithm, we consider two cases.
    \begin{itemize}
    \item If $v\not\in C$ then we have
          \[
          \Delta(v,C) = d(v,q_0) + d(P,u(P,C)) 
                 = d(v,q_0) + \max_{u\in C} d(P,u)
                 = \max_{u\in C} \hspace{0.2mm} \big\{ d(v,q_0) + d(P,u) \big\}.
          \]
          By \autoref{lem: approx-path-boundary}(\ref{approx-i}) this implies $\max_{u\in C} d(u,v)\leq \Delta(v,C)$.
          Moreover, because $v\not\in C$, there is a node~$s\in S(C)$ that is either a node
          on a shortest path from $v$ to $u(P,C)$ or the neighbor of such a node,
          by \autoref{lem:inside-to-outside-cluster}. 
          By \autoref{lem: approx-path-boundary}(\ref{approx-ii}) and~(\ref{approx-iii})
          we therefore have 
          \[
          \Delta(v,C) = d(v,q_0) + d(P,u(P,C)) \leq d(v,u(P,C)) + 2 \leq \max_{u\in C} d(u,v) + 2
          \]          
    \item If $v\in C$ then it suffices to argue that for all $u\in C$ we have
        \[
        d(u,v) \leq \min\{d_{\ig[C]}(u,v), \ d( v,q_0)+d(P,u) \}\leq d(u,v) + 2.
        \]
        Clearly, $d(u,v) \leq d_{\ig[C]}(u,v)$. Moreover,
        $d(u,v) \leq d(v,q_0)+d(P,u)$ by \autoref{lem: approx-path-boundary}(\ref{approx-i}).
        We conclude that the first inequality holds.
        
        If there is a shortest path $\pi(u,v)$ that stays inside $C$ then
        $d_{\ig[C]}(u,v)=d(u,v)$. To prove the second inequality, it therefore remains to show that 
        $d(v,q_0)+d(P,u) \leq d(u,v)+2$ when there is a shortest path $\pi(u,v)$
        that visits a node $w\not\in C$. By \autoref{lem:inside-to-outside-cluster}
        we then have a node $s \in S(C)$ that is in $\pi(u,v)$ or has a neighbor in $\pi(u,v)$.
        Similarly to before, we can use \autoref{lem: approx-path-boundary}(\ref{approx-ii})
        to conclude $d(v,q_0) + d(P,u) \leq d(v,u) + 2$.
    \end{itemize}

    \noindent It remains to show that the algorithm runs in $\tilde{O}(n^{2-1/14})$ time.
    \begin{itemize}
    \item Step~\ref{step:diam1} first computes a star-based $r$-clustering $\C$,
        which can be done in $\tilde{O}(n)$ time by \autoref{thm: star-based-clustering}. 
        We then run our \sssp algorithm from all star centers, of which there are $O(n/\sqrt{r})$. It follows from
        \autoref{thm:sssp} that this takes
        $\tilde{O}(n^2/\sqrt{r})$ time in total. Generating the patterns $\pt$ for all pairs $v,C$
        can be done by looking up 
        the computed distances. This takes $\sum_{v\in V} \sum_{C\in \C} |S(C)| =O(n^2/\sqrt{r})$ time,
        because $\sum_{C\in\C} |S(C)|= O(n/\sqrt{r})$ by our bound on the global boundary complexity.
        Moreover, removing duplicates
        by sorting has only a logarithmic overhead. (Doing a comparison between two patterns with respect to~$Q(C)$
        takes $\Theta(|Q(C)|)$ time, but the total number of patterns is a factor $|Q(C)|$ smaller
        than the sum of their sizes.) Finally,
          computing all distances $d_{\ig[C]}(u,v)$ takes $\tilde{O}(\sum_{C\in\C} |C|^2) = \tilde{O}(r^2 |\C|) = \tilde{O}(n r\sqrt{r})$
          expected time, since we run our \sssp algorithm $|C|$~times.
          Hence, Step~\ref{step:diam1}
        takes $\tilde{O}(n^2/\sqrt{r})$ time.
    \item 
        The total time spent in Step~\ref{step:diam2} is bounded by
        \[
        O\left( \sum\limits_{C\in \C} \sum\limits_{P\in \cP_C} \sum\limits_{u\in C} \mbox{length of pattern $P$} \right)
          =  O\left( \sum\limits_{C\in \C} (\mbox{\# patterns in $C$}) \cdot |C| \cdot |S(C)| \right).
        \]
        Because each cluster is connected, the distance between any two nodes in $Q(C)$
        is at most~$|C|$. Hence, \autoref{lem: pattern} and the bound from \autoref{thm: star-based-clustering}
        on the local boundary complexity together imply that
        \[
        \mbox{\# patterns in $C$} = O((|C| \cdot |S(C)|)^4) = O( (r\sqrt{r})^4) = O(r^6).
        \]
        Recall that the size $|C|$ of a cluster is at most~$r$ and that
        $\sum_{C\in\C}|S(C)|$, the global boundary complexity of our $r$-clustering,
        is $O(n/\sqrt{r})$. Hence,
        the total running time for~ Step~\ref{step:diam2} becomes
        \[
        O\left( \sum\limits_{C\in \C} r^6 \cdot r \cdot |S(C)| \right) = O\left( r^7 \cdot (n/\sqrt{r}) \right).
        \]
    \item To analyze Step~\ref{step:diam3}, we first observe that the time needed 
          for each node $v\in V$ to compute $\Delta(v,C)$ in Case~(a)
          is dominated by the time we already spent in earlier steps.
          Since we have the values $d(P,u)$ available, computing all
          the values $d(u,q_0)+d(P,u)$ in Case~(b) takes $O(|C|^2)$ for each cluster. Hence, the total expected time for Case~(b)
          is $\tilde{O}(\sum_{C\in\C} |C|^2) = \tilde{O}(r^2 |\C|) = \tilde{O}(n r\sqrt{r})$.
    \end{itemize}
    Summing the bounds for the three steps, we see that the expected total running time
    is $\tilde{O}(n^2/\sqrt{r}+r^7 (n/\sqrt{r}))$. Setting $r=n^{1/7}$ we get the 
    claimed running time of $\tilde{O}(n^{2-1/14})$.
\end{proof}
The techniques used to obtain \autoref{thm:diam} can also be used to construct a distance oracle,
similarly to what is done by Chang, Gao, and Le~\cite{Chang0024} in their setting.
Like for \diamp, we not only obtain a more general result but we also get a better bound,
in this case on the amount of storage. Somewhat surprisingly, the storage bound
for our oracle is smaller than the running time of our \diamp algorithm.
\begin{corollary}
    Let $D$ be a set of $n$ constant-complexity pseudodisks. Then we can construct 
    in $\tilde{O}(n^2)$ expected time a distance oracle for $\ig[D]$  that uses $O(n^{2-1/13})$ storage, 
    has query time $O(1)$, and returns the distance between two query nodes up to an additive error of 2.  
\end{corollary}
\begin{proof}
The oracle is constructed as follows.

We first run Step~\ref{step:diam1} of our \diamp algorithm to compute a star-based $r$-clustering $\C$ of~$\G := \ig[\D]$
as well as, for each cluster~$C\in\C$, the sequence~$Q(C)$ of star centers and the set~$\cP_C$ of patterns.
This time, however, we use a different value of $r$, as specified later.
Recall that each node $u\in V$ can be present in several clusters. We designate one of these
clusters as the \emph{home cluster} of $u$, denoted by~$\home{u}$. Distinguishing home
clusters allows us reduce the storage of our oracle, compared to the running time of our \diamp algorithm.
Our distance oracle now stores the following information.
\begin{enumerate}
    \item For each cluster $C \in \C$ and each $u,v \in C \times C$, we store $d_{\ig[C]}(u,v)$.
    \item For each cluster $C \in \C$, each node $u \in C$ such that $\home{u}=C$, and each pattern $P \in \cP_C$, we store $d(P,u)$.
          Note that we do not explicitly store the patterns $P$, which can have size $\Theta(\sqrt{r})$; a label identifying each pattern suffices.
    \item For each node $u\in V$ and each cluster $C \in \C$, we store a pointer from~$u$ to the pattern $P\in \cP_C$ 
    such that $P =\pt[u,Q(C)]$, and we store $d(u,q_0)$ where $q_0$ is the first element of~$Q(C)$.
\end{enumerate}
The oracle answers a distance query with nodes $u,v\in V$ as follows.
Let $C:= \home{v}$, let $q_0$ be the first element of~$Q(C)$, and let $P := \pt[u,Q(C)]$
be the pattern of $u$ with respect to~$Q(C)$.
\begin{itemize}
    \item If $u\not \in C$, we return $d(u,q_0)+d(P,v)$. 
    \item If $u \in C$, we  return $\min\{d_{\ig[C]}(u,v), d(u,q_0)+d(P,v) \}$.
\end{itemize} 

The values needed to answer the query have all been precomputed, thus a query takes $O(1)$ time.
The correctness proof of the query is similar to the correctness proof of the \diamp algorithm.
Indeed, if the shortest path from $u$ to $v$ visits a node outside~$C$
then we can use \autoref{lem:inside-to-outside-cluster}
and \autoref{lem: approx-path-boundary}(\ref{approx-ii}) to conclude that
$d(u,q_0)+d(P,v)\leq d(u,v)+2$, which is easily seen to imply that we always report
the distance up to an additive error of at most~2.
\medskip

It remains to analyze the preprocessing time and the amount of storage.
It follows from the analysis of the \diamp algorithm that we need $\tilde{O}(n^2 /\sqrt{r}+r^7(n/\sqrt{r}))$ expected time 
to construct the oracle. 
As for the storage, observe that we require $O(\sum_{C\in\C} |C|^2)=O(n r\sqrt{r})$ storage for the first part of the oracle. 
For the second part we require $O(nr^6)$ storage, since there are $O(r^6)$ patterns per cluster and each vertex has only one home cluster. 
This is less than the computation time needed for the second part, because
we do not explicitly store each pattern $P \in \cP_C$ and we only store $d(P,u)$ for nodes $u$ with $\home{u}=C$. 
For the third part, we require $O(n^2 / \sqrt{r})$ storage.
In total we need $O(n^2 /\sqrt{r}+nr^6)$ storage. This is minimized when we set $r:=n^{2/13}$,
leading to a storage bound of~$O(n^{2-1/13})$.
The preprocessing time then becomes~$\tilde{O}(n^2)$.
\end{proof}

\section{Concluding remarks}
We presented the first $O(n\polylog n)$ \sssp algorithm for arbitrary (constant-complexity) pseudodisks. 
This significantly extends the class of objects for which \sssp can be solved in near-linear time:
so far, the only $O(n\polylog n)$ \sssp algorithm for pseudodisks required them to
be similarly-sized, fat, and satisfy a strong monotonicity assumption,  or 
(for arbitrarily-sized objects) be a Euclidean disk.
Our algorithm not only works for any arbitrary set of pseudodisks, but it works
for any set~$\D$ of constant-complexity objects whose union complexity is $O(n\polylog n)$.
More precisely, the expected running time of our algorithm is $O(\lambda_{s+2}(U(n))\cdot \log^2 n)$,
where $U(n)$ denotes the union complexity of $\D$ and $s$ is the
maximum number of intersections between any pair of boundary arcs.

An obvious question is whether the running time can be improved. 
We note that until recently~\cite{DBLP:conf/esa/Brewer025,DBLP:conf/esa/BergC25} the fastest
\sssp algorithm for Euclidean disks ran in $O(n\log^2 n)$ time, and that these recent improvements
for disks use techniques that do not apply in our much more general setting.
Hence, improving the $O(\lambda_{s+2}(U(n))\cdot \log^2 n)$ running time that we achieve might
be quite challenging.

Apart from log-shaving, one can also wonder if there are more general families
of objects that admit an $O(n\polylog n)$ \sssp algorithm. This
is most likely not possible for segments, because of Erickson's $\Omega(n^{4/3})$ 
lower bound~\cite{DBLP:journals/dcg/Erickson96} for Hopcroft's problem. 
We do not see any natural class of objects that is more general than 
objects with near-linear union complexity but less general than segments, 
so it is not quite clear what to look for.
\medskip

We also presented an efficient \emph{star-based $r$-clustering} for objects with
small union complexity; this is a variant of widely used $r$-divisions of 
planar graphs~\cite{DBLP:journals/siamcomp/Frederickson87}. We expect that 
our star-based clusterings may find other applications as well.
\medskip

Finally, we showed that our \sssp algorithm and star-based $r$-clusterings
together are enough to run the \diamp algorithm of Chang, Gao, and Le~\cite{Chang0024}
to approximate the diameter of a pseudodisk graph up to an additive factor of~2.
This not only extends the class of objects for which the diameter can be
computed almost exactly, but the stronger complexity guarantees of our
star-based $r$-clustering also mean that the $\tilde{O}(n^{2-1/14})$
running time that we achieve is better than what Chang, Gao, and Le obtain.
Can the running time be improved even further? Or can we extend the result
from pseudodisks to objects with near-linear union complexity? For computing
the exact diameter, this is ruled out (under SETH) by the lower bound of
Bringmann~\etal~\cite{DBLP:conf/compgeom/BringmannKKNP22}, but does a subquadratic
algorithm exist that is almost exact? 
And can we perhaps even
obtain an exact algorithm with subquadratic running time for pseudodisks? Since this is so
far only known for unit disks and squares, a more modest goal would be to
obtain a subquadratic exact algorithm for (arbitrarily-sized) disks.

\bibliography{references}

@article{DBLP:journals/algorithmica/BaschGR03,
  author       = {Julien Basch and
                  Leonidas J. Guibas and
                  G. D. Ramkumar},
  title        = {Reporting Red - Blue Intersections between Two Sets of Connected Line
                  Segments},
  journal      = {Algorithmica},
  volume       = {35},
  number       = {1},
  pages        = {1--20},
  year         = {2003},
  doi          = {10.1007/S00453-002-0967-4}
}

@inproceedings{DBLP:conf/esa/BergC25,
  author       = {Mark de {Berg} and
                  Sergio Cabello},
  title        = {An ${O}(n\log n)$ Algorithm for Single-Source Shortest Paths in Disk Graphs},
  booktitle    = {Proc.~33rd Annual European Symposium on Algorithms ({ESA})},
  series       = {LIPIcs},
  volume       = {351},
  pages        = {81:1--81:15},
  year         = {2025},
  doi          = {10.4230/LIPICS.ESA.2025.81}
}

@article{DBLP:journals/dcg/Har-PeledS01,
  author       = {Sariel Har{-}Peled and
                  Micha Sharir},
  title        = {Online Point Location in Planar Arrangements and Its Applications},
  journal      = {Discret. Comput. Geom.},
  volume       = {26},
  number       = {1},
  pages        = {19--40},
  year         = {2001},
  doi          = {10.1007/S00454-001-0026-Y}
}

@book{DBLP:books/daglib/0080837,
  author       = {Micha Sharir and
                  Pankaj K. Agarwal},
  title        = {Davenport-Schinzel sequences and their geometric applications},
  publisher    = {Cambridge University Press},
  year         = {1995},
  isbn         = {978-0-521-47025-4}
}

@article{Szemeredi-DS,
  author       = {Endre Szemer\'edi},
  title        = {On a problem of {Davenport} and {Schinzel}},
  journal      = {Acta Arith.},
  volume       = {25},
  pages        = {213--224},
  year         = {1973}
}

@article{doi:10.1137/20M1388371,
author = {Liu, Chih-Hung},
title = {Nearly Optimal Planar $k$ Nearest Neighbors Queries under General Distance Functions},
journal = {SIAM Journal on Computing},
volume = {51},
number = {3},
pages = {723-765},
year = {2022},
doi = {10.1137/20M1388371}
}

@article{DBLP:journals/cpc/Sharir03,
  author       = {Micha Sharir},
  title        = {The Clarkson-Shor Technique Revisited And Extended},
  journal      = {Comb. Probab. Comput.},
  volume       = {12},
  number       = {2},
  pages        = {191--201},
  year         = {2003},
  doi          = {10.1017/S0963548302005527}
}

@article{DBLP:journals/dcg/ClarksonS89,
  author       = {Kenneth L. Clarkson and
                  Peter W. Shor},
  title        = {Application of Random Sampling in Computational Geometry, {II}},
  journal      = {Discret. Comput. Geom.},
  volume       = {4},
  pages        = {387--421},
  year         = {1989},
  doi          = {10.1007/BF02187740}
}

@article{DBLP:journals/jacm/Mulmuley91,
  author       = {Ketan Mulmuley},
  title        = {A Fast Planar Partition Algorithm, {II}},
  journal      = {J. {ACM}},
  volume       = {38},
  number       = {1},
  pages        = {74--103},
  year         = {1991},
  doi          = {10.1145/102782.102785}
}

@book{DBLP:books/daglib/0077673,
  author       = {Ketan Mulmuley},
  title        = {Computational geometry - an introduction through randomized algorithms},
  publisher    = {Prentice Hall},
  year         = {1994}
}

@article{DBLP:journals/dcg/BoissonnatDSTY92,
  author       = {Jean{-}Daniel Boissonnat and
                  Olivier Devillers and
                  Ren{\'{e}} Schott and
                  Monique Teillaud and
                  Mariette Yvinec},
  title        = {Applications of Random Sampling to On-line Algorithms in Computational
                  Geometry},
  journal      = {Discret. Comput. Geom.},
  volume       = {8},
  pages        = {51--71},
  year         = {1992},
  doi          = {10.1007/BF02293035}
}

@article{KaplanMRSS20,
  author       = {Haim Kaplan and Wolfgang Mulzer and Liam Roditty and Paul Seiferth and Micha Sharir},
  title        = {Dynamic Planar {V}oronoi Diagrams for General Distance Functions and Their Algorithmic Applications},
  journal      = {Discret. Comput. Geom.},
  volume       = {64},
  number       = {3},
  pages        = {838--904},
  year         = {2020},
  doi          = {10.1007/S00454-020-00243-7},
}

@article{DBLP:journals/dcg/Erickson96,
  author       = {Jeff Erickson},
  title        = {New Lower Bounds for {H}opcroft's Problem},
  journal      = {Discret. Comput. Geom.},
  volume       = {16},
  number       = {4},
  pages        = {389--418},
  year         = {1996},
  doi          = {10.1007/BF02712875}
}

@article{AronovBES14,
  author       = {Boris Aronov and Mark de Berg and Esther Ezra and Micha Sharir},
  title        = {Improved Bounds for the Union of Locally Fat Objects in the Plane},
  journal      = {{SIAM} J. Comput.},
  volume       = {43},
  number       = {2},
  pages        = {543--572},
  year         = {2014},
  doi          = {10.1137/120891241}
}

@inproceedings{DBLP:conf/esa/Brewer025,
  author       = {Bruce W. Brewer and
                  Haitao Wang},
  title        = {An Optimal Algorithm for Shortest Paths in Unweighted Disk Graphs},
  booktitle    = {Proc.~33rd Annual European Symposium on Algorithms ({ESA})},  
  series       = {LIPIcs},
  volume       = {351},
  pages        = {31:1--31:8},
  year         = {2025},
  doi          = {10.4230/LIPICS.ESA.2025.31}
}

@article{CabelloJ15,
  author       = {Sergio Cabello and Miha Jej{\v c}i{\v c}},
  title        = {Shortest paths in intersection graphs of unit disks},
  journal      = {Comput. Geom.},
  volume       = {48},
  number       = {4},
  pages        = {360--367},
  year         = {2015},
  doi          = {10.1016/j.comgeo.2014.12.003},
}

@article{EfratIK01,
  author       = {Alon Efrat and Alon Itai and Matthew J. Katz},
  title        = {Geometry Helps in Bottleneck Matching and Related Problems},
  journal      = {Algorithmica},
  volume       = {31},
  number       = {1},
  pages        = {1--28},
  year         = {2001},
  doi          = {10.1007/S00453-001-0016-8},
}

@article{ChanS19,
  author       = {Timothy M. Chan and Dimitrios Skrepetos},
  title        = {All-Pairs Shortest Paths in Geometric Intersection Graphs},
  journal      = {J. Comput. Geom.},
  volume       = {10},
  number       = {1},
  pages        = {27--41},
  year         = {2019},
  doi          = {10.20382/JOCG.V10I1A2},
}

@inproceedings{ChanS16,
  author       = {Timothy M. Chan and Dimitrios Skrepetos},
  editor       = {Seok{-}Hee Hong},
  title        = {All-Pairs Shortest Paths in Unit-Disk Graphs in Slightly Subquadratic Time},
  booktitle    = {Proc.~27th International Symposium on Algorithms and Computation (ISAAC)},
  series       = {LIPIcs},
  volume       = {64},
  pages        = {24:1--24:13},
  year         = {2016},
  doi          = {10.4230/LIPICS.ISAAC.2016.24},
}

@inproceedings{Chang0024,
  author       = {Hsien{-}Chih Chang and Jie Gao and Hung Le},
  title        = {Computing Diameter+2 in Truly-Subquadratic Time for Unit-Disk Graphs},
  booktitle    = {Proc. 40th International Symposium on Computational Geometry ({SoCG})},
  series       = {LIPIcs},
  volume       = {293},
  pages        = {38:1--38:14},
  year         = {2024},
  doi          = {10.4230/LIPICS.SoCG.2024.38}
}

@inproceedings{DBLP:conf/focs/ChanCGKLZ25,
  author       = {Timothy M. Chan and
                  Hsien{-}Chih Chang and
                  Jie Gao and
                  S{\'{a}}ndor Kisfaludi{-}Bak and
                  Hung Le and
                  Da Wei Zheng},
  title        = {Truly Subquadratic Time Algorithms for Diameter and Related Problems
                  in Graphs of Bounded VC-dimension},
  booktitle    = {Proc.~66th {IEEE} Annual Symposium on Foundations of Computer Science ({FOCS})},
  pages        = {2728--2765},
  year         = {2025},
  doi          = {10.1109/FOCS63196.2025.00140}
}

@book{bcko-cgaa-08,
  author    = {Mark de {B}erg and
               Otfried Cheong and
               Marc J. van Kreveld and
               Mark H. Overmars},
  title     = {Computational Geometry: Algorithms and Applications (3rd Edition)},
  publisher = {Springer},
  year      = {2008},
  doi       = {10.1007/978-3-540-77974-2},
}

@article{DBLP:journals/siamcomp/Frederickson87,
  author       = {Greg N. Frederickson},
  title        = {Fast Algorithms for Shortest Paths in Planar Graphs, with Applications},
  journal      = {{SIAM} J. Comput.},
  volume       = {16},
  number       = {6},
  pages        = {1004--1022},
  year         = {1987},
  doi          = {10.1137/0216064}
}

@article{Klost23,
  author       = {Katharina Klost},
  title        = {An algorithmic framework for the single source shortest path problem with applications to disk graphs},
  journal      = {Comput. Geom.},
  volume       = {111},
  pages        = {101979},
  year         = {2023},
  doi          = {10.1016/j.comgeo.2022.101979}
}

@article{WULFFNILSEN2013831,
author = {Christian Wulff-Nilsen},
title = {Constant time distance queries in planar unweighted graphs with subquadratic preprocessing time},
journal = {Computational Geometry},
volume = {46},
number = {7},
pages = {831-838},
year = {2013},
doi = {doi.org/10.1016/j.comgeo.2012.01.016}
}

@article{DBLP:journals/csur/Sommer14,
  author       = {Christian Sommer},
  title        = {Shortest-path queries in static networks},
  journal      = {{ACM} Comput. Surv.},
  volume       = {46},
  number       = {4},
  pages        = {45:1--45:31},
  year         = {2014},
  doi          = {10.1145/2530531}
}

@article{DBLP:journals/jacm/CharalampopoulosGLMPWW23,
  author       = {Panagiotis Charalampopoulos and
                  Pawel Gawrychowski and
                  Yaowei Long and
                  Shay Mozes and
                  Seth Pettie and
                  Oren Weimann and
                  Christian Wulff{-}Nilsen},
  title        = {Almost Optimal Exact Distance Oracles for Planar Graphs},
  journal      = {J. {ACM}},
  volume       = {70},
  number       = {2},
  pages        = {12:1--12:50},
  year         = {2023},
  doi          = {10.1145/3580474}
}

@article{DBLP:journals/dcg/Berg08,
  author       = {Mark de {Berg}},
  title        = {Improved Bounds on the Union Complexity of Fat Objects},
  journal      = {Discret. Comput. Geom.},
  volume       = {40},
  number       = {1},
  pages        = {127--140},
  year         = {2008},
  doi          = {10.1007/S00454-007-9029-7}
}

@article{DBLP:journals/talg/Cabello19,
  author       = {Sergio Cabello},
  title        = {Subquadratic Algorithms for the Diameter and the Sum of Pairwise Distances
                  in Planar Graphs},
  journal      = {{ACM} Trans. Algorithms},
  volume       = {15},
  number       = {2},
  pages        = {21:1--21:38},
  year         = {2019},
  doi          = {10.1145/3218821}
}

@inproceedings{DBLP:conf/compgeom/BringmannKKNP22,
  author       = {Karl Bringmann and
                  S{\'{a}}ndor Kisfaludi{-}Bak and
                  Marvin K{\"{u}}nnemann and
                  Andr{\'{e}} Nusser and
                  Zahra Parsaeian},
  title        = {Towards Sub-Quadratic Diameter Computation in Geometric Intersection
                  Graphs},
  booktitle    = {Proc.~38th International Symposium on Computational Geometry (SoCG)},
  series       = {LIPIcs},
  volume       = {224},
  pages        = {21:1--21:16},
  year         = {2022},
  doi          = {10.4230/LIPICS.SOCG.2022.21}
}

@article{DBLP:journals/siamcomp/GawrychowskiKMS21,
  author       = {Pawel Gawrychowski and
                  Haim Kaplan and
                  Shay Mozes and
                  Micha Sharir and
                  Oren Weimann},
  title        = {Voronoi Diagrams on Planar Graphs, and Computing the Diameter in Deterministic
                  {\~{O}}(n\({}^{\mbox{5/3}}\)) Time},
  journal      = {{SIAM} J. Comput.},
  volume       = {50},
  number       = {2},
  pages        = {509--554},
  year         = {2021},
  doi          = {10.1137/18M1193402}
}

@article{DBLP:journals/jacm/Pettie15,
  author       = {Seth Pettie},
  title        = {Sharp Bounds on Davenport-Schinzel Sequences of Every Order},
  journal      = {J. {ACM}},
  volume       = {62},
  number       = {5},
  pages        = {36:1--36:40},
  year         = {2015},
  doi          = {10.1145/2794075}
}

@article{gao-zhang,
author = {Gao, Jie and Zhang, Li},
title = {Well-Separated Pair Decomposition for the Unit-Disk Graph Metric and Its Applications},
journal = {SIAM Journal on Computing},
volume = {35},
number = {1},
pages = {151-169},
year = {2005},
doi = {10.1137/S0097539703436357}
}

@InProceedings{chan-skrepetos,
  author =	{Chan, Timothy M. and Skrepetos, Dimitrios},
  title =	{{Approximate Shortest Paths and Distance Oracles in Weighted Unit-Disk Graphs}},
  booktitle =	{Proc.~34th International Symposium on Computational Geometry (SoCG)},
  pages =	{24:1--24:13},
  year =	{2018},
  volume =	{99},
  doi =		{10.4230/LIPIcs.SoCG.2018.24}
}

@inproceedings{DBLP:conf/isaac/BergJL25,
  author       = {Mark de Berg and
                  Bart M. P. Jansen and
                  Jeroen S. K. Lamme},
  editor       = {Ho{-}Lin Chen and
                  Wing{-}Kai Hon and
                  Meng{-}Tsung Tsai},
  title        = {Star-Based Separators for Intersection Graphs of c-Colored Pseudo-Segments},
  booktitle    = {Proc.~36th International Symposium on Algorithms and Computation ({ISAAC})},
  series       = {LIPIcs},
  pages        = {12:1--12:16},
  year         = {2025},
  doi          = {10.4230/LIPICS.ISAAC.2025.12}
}

@article{DBLP:journals/algorithmica/AronovBT25,
  author       = {Boris Aronov and
                  Mark de Berg and
                  Leonidas Theocharous},
  title        = {A Clique-Based Separator for Intersection Graphs of Geodesic Disks
                  in {\textdollar}{\textbackslash}mathbb \{R\}{\^{}}2{\textdollar}},
  journal      = {Algorithmica},
  volume       = {87},
  number       = {12},
  pages        = {1997--2017},
  year         = {2025},
  doi          = {10.1007/S00453-025-01337-5}
}

@article{DBLP:journals/corr/abs-2401-12881,
  author       = {Hsien{-}Chih Chang and
                  Jie Gao and
                  Hung Le},
  title        = {Computing Diameter+1 in Truly Subquadratic Time for Unit-Disk Graphs},
  journal      = {CoRR},
  volume       = {abs/2401.12881},
  year         = {2024},
  doi          = {10.48550/ARXIV.2401.12881},
  eprinttype    = {arXiv}
}

@article{pseudodisk-union-86,
  author    = {Klara Kedem and
               Ron Livne and
               J{\'{a}}nos Pach and
               Micha Sharir},
  title     = {On the Union of Jordan Regions and Collision-Free Translational Motion
               Amidst Polygonal Obstacles},
  journal   = {Discret. Comput. Geom.},
  volume    = {1},
  pages     = {59--70},
  year      = {1986},
  doi       = {10.1007/BF02187683}
}

@inproceedings{DBLP:conf/isaac/Fredslund-Hansen21,
  author       = {Viktor Fredslund{-}Hansen and
                  Shay Mozes and
                  Christian Wulff{-}Nilsen},
  title        = {Truly Subquadratic Exact Distance Oracles with Constant Query Time
                  for Planar Graphs},
  booktitle    = {Proc.~32nd International Symposium on Algorithms and Computation ({ISAAC})},
  volume       = {212},
  pages        = {25:1--25:12},
  year         = {2021},
  doi          = {10.4230/LIPICS.ISAAC.2021.25}
}

@inproceedings{DBLP:conf/stoc/LiP19,
  author       = {Jason Li and
                  Merav Parter},
  title        = {Planar diameter via metric compression},
  booktitle    = {Proceedings of the 51st Annual {ACM} {SIGACT} Symposium on Theory
                  of Computing ({STOC})},
  pages        = {152--163},
  year         = {2019},
  doi          = {10.1145/3313276.3316358}
}

@article{DBLP:journals/jacm/ThorupZ05,
  author       = {Mikkel Thorup and
                  Uri Zwick},
  title        = {Approximate distance oracles},
  journal      = {J. {ACM}},
  volume       = {52},
  number       = {1},
  pages        = {1--24},
  year         = {2005},
  doi          = {10.1145/1044731.1044732}
}

\end{document}